\documentclass{article}

\usepackage{microtype}
\usepackage{graphicx}
\usepackage{subfigure}
\usepackage{booktabs} % for professional tables
\usepackage{comment}

\usepackage{hyperref}

\usepackage[accepted]{icml2024}

\usepackage{amsmath}
\usepackage{amssymb}
\usepackage{mathtools}
\usepackage{amsthm}

\usepackage[capitalize,noabbrev]{cleveref}

\theoremstyle{plain}
\newtheorem{theorem}{Theorem}[section]
\newtheorem{proposition}[theorem]{Proposition}
\newtheorem{lemma}[theorem]{Lemma}
\newtheorem{corollary}[theorem]{Corollary}
\theoremstyle{definition}
\newtheorem{definition}[theorem]{Definition}
\newtheorem{assumption}[theorem]{Assumption}
\theoremstyle{remark}
\newtheorem{remark}[theorem]{Remark}

\usepackage[textsize=tiny]{todonotes}

\DeclareMathOperator*{\argmax}{arg\,max}
\DeclareMathOperator*{\argmin}{arg\,min}

\renewcommand{\eqref}[1]{(\ref{#1})}

\usepackage{enumitem}
\usepackage{bm}
\usepackage{bbm}
\allowdisplaybreaks

\newcommand{\iid}{ \stackrel{\mathrm{i.i.d}}{\sim} }

\usepackage{url}

\usepackage{xcolor}
\newcount\Comments  % 0 suppresses notes to selves in text
\newcommand{\kibitz}[2]{\ifnum\Comments=1\textcolor{#1}{#2}\fi}

\icmltitlerunning{\emph{Active} Adaptive Experimental Design for Treatment Effect Estimation with Covariate Choice}

\begin{document}

\twocolumn[
\icmltitle{\emph{Active} Adaptive Experimental Design for Treatment Effect Estimation\\
with Covariate Choice}




\begin{icmlauthorlist}
\icmlauthor{Masahiro Kato}{mizuho}
\icmlauthor{Akihiro Oga}{mizuho}
\icmlauthor{Wataru Komatsubara}{mizuho}
\icmlauthor{Ryo Inokuchi}{mizuho}
\end{icmlauthorlist}

\icmlaffiliation{mizuho}{Mizuho-DL Financial Technology Co., Ltd., Tokyo, Japan}

\icmlcorrespondingauthor{Masahiro Kato}{masahiro-kato@fintec.co.jp}

\icmlkeywords{Treatment effect estimation, Adaptive experimental design, Active learning, Causal inference}

\vskip 0.3in
]



\printAffiliationsAndNotice{} % otherwise use the standard text.

\begin{abstract}
This study designs an adaptive experiment for efficiently estimating \emph{average treatment effects} (ATEs). In each round of our adaptive experiment, an experimenter sequentially samples an experimental unit, assigns a treatment, and observes the corresponding outcome immediately. At the end of the experiment, the experimenter estimates an ATE using the gathered samples. The objective is to estimate the ATE with a smaller asymptotic variance. Existing studies have designed experiments that adaptively optimize the propensity score (treatment-assignment probability). As a generalization of such an approach, we propose optimizing the covariate density as well as the propensity score. First, we derive the efficient covariate density and propensity score that minimize the semiparametric efficiency bound and find that optimizing both covariate density and propensity score minimizes the semiparametric efficiency bound more effectively than optimizing only the propensity score. Next, we design an adaptive experiment using the efficient covariate density and propensity score sequentially estimated during the experiment. Lastly, we propose an ATE estimator whose asymptotic variance aligns with the minimized semiparametric efficiency bound.
\end{abstract}

\section{Introduction}
\emph{Experimental approaches} play a pivotal role in uncovering causality in various fields of science and industrial applications, such as epidemiology and economics \citep{ChowChangPong,Hahn2011,ChowChang201112,fda}. This study focuses on cases with \emph{binary treatments}. Typically, one of these treatments corresponds to the \emph{treatment}, and the other to the \emph{control} \citep{imbens_rubin_2015}. Among the quantities representing causal effects, this study focuses on an \emph{average treatment effect} (ATE), which defines a causal effect as the difference between the expected outcomes of binary treatments.

\emph{Treatments} are an abstraction of an experimenter's choices, including drugs, online advertisements, and economic policies, which are also referred to differently, such as actions and arms. A difference in outcomes of binary treatments (treatment effect) is often considered a causal effect \citep{imbens_rubin_2015}. However, due to the counterfactual property, \emph{individual treatment effects} cannot be observed. Therefore, we use ATEs to capture causality.

A fundamental experimental approach for estimating ATEs is a \emph{randomized control trial} (RCT). In an RCT, one of the binary treatments is randomly assigned to each experimental unit to obtain an unbiased estimator of ATEs. While RCTs are the gold standard and reliable \citep{Hariton2018}, they often require large sample sizes, which can be costly. To reduce the sample size as much as possible, or to reduce an estimation error given a sample size, \emph{adaptive experimental designs} have gained increasing attention across various fields \citep{Laan2008TheCA,Hahn2011,ChowChang201112,villar,fda}. 

For estimating ATEs efficiently (with smaller asymptotic variances), \citet{Laan2008TheCA} and \citet{Hahn2011} design adaptive experiments that optimize the propensity score (treatment-assignment probability). They show that by optimizing the propensity score, we can minimize the asymptotic variance of ATE estimators. Note that an asymptotic variance can be interpreted as an asymptotic mean squared error (MSE) of ATE estimation. Furthermore, in hypothesis testing, as an asymptotic variance decreases, the sample size required for hypothesis testing also decreases. \citet{Meehan2022}, \citet{Kato2020adaptive}, and \citet{cook2023semiparametric} design refined experiments based on this approach.

As a generalization of this approach, we consider a novel setting where an experimenter can sample experimental units based on their covariates and covariate density set by the experimenter. We show that by optimizing the covariate density as well as the propensity score, we can further minimize asymptotic variances compared to adaptive experiments that only optimize the propensity score. Then, we design an adaptive experiment that optimizes efficient covariate density and propensity score by using past observations.

Firstly, we derive the semiparametric efficiency bound, a lower bound of the asymptotic variance, which is a functional of a covariate density and a propensity score. Based on this efficiency bound, we compute the efficient covariate density and propensity score that minimize the lower bound. Then, we design an adaptive experiment that sequentially estimates the efficient covariate density and propensity score. In the proposed experiment, an experimental unit is sampled following an adaptively estimated efficient covariate density and randomly assigned treatments following an estimated efficient propensity score. At the experiment's end, the experimenter obtains an estimator whose asymptotic variance aligns with the semiparametric efficiency bound. We employ the \emph{Augmented Inverse Probability Weighting} (AIPW) estimator with covariate shift adaptation using Importance Weighting (AIPWIW). We refer to the setting of an adaptive experiment as \emph{active adaptive experimental design with covariate choice} and our designed experiment as the \emph{Active-Adaptive-Sampling (AAS)-AIPWIW experiment}.

Our method is inspired by three lines of work: adaptive experimental design for efficient ATE estimation, active learning using a covariate shift, and off-policy evaluation under a covariate shift. \citet{Laan2008TheCA}, \citet{Hahn2011}, \citet{Meehan2022}, and \citet{Kato2020adaptive} design adaptive experiments for estimating ATEs with smaller asymptotic variances by optimizing propensity scores. Our result generalizes their work by incorporating covariate density optimization. In active learning of regression models, \citet{Sugiyama2006} demonstrates that the asymptotic variance of regression coefficient estimators can be reduced by appropriately shifting the covariate density. That study deals with a non-adaptive experiment and regression with importance weighting \citep{Shimodaira2000}. In contrast, our target is ATE estimation, and we adaptively optimize the covariate density to estimate ATEs. Note that this approach for active learning is an application of importance sampling \citep{mcbook}. In causal inference under a covariate shift, \citet{Uehara2020} presents an efficient estimator of policy values, a generalization of ATEs, using double machine learning \citep[DML,][]{ChernozhukovVictor2018Dmlf}.

Intuitively, we can reduce the asymptotic variance of the ATE estimator by using more units for a treatment with a larger variance. In existing studies, an experimenter assigns a treatment to an experimental unit with a higher propensity score if it has a larger variance. In this study, we also optimize the choice of experimental units based on their covariates in addition to the propensity score. Under this approach, we use more experimental units if their conditional variances of outcomes are large, conditioned on their covariates. See Figure~\ref{fig:concept_figure2} in Section~\ref{sec:efficiency_prob}. This covariate choice allows us to reduce the asymptotic variance of an ATE estimator.

Our framework introduces a novel aspect of optimizing covariate density, which has not been fully explored in the literature on adaptive experiments for ATE estimation. While active learning focuses on optimizing the covariate density, adaptive experiments typically focus on optimizing the propensity score. We refer to our setting as an \emph{active adaptive experiment} for ATE estimation.

\paragraph{Contributions} Compared to existing work, our contributions are summarized as follows:
\begin{itemize}[topsep=0pt, itemsep=0pt, partopsep=0pt, leftmargin=*]
    \item We propose a novel framework for an adaptive experiment by optimizing the covariate density, as well as the propensity score (Section~\ref{sec:problem});
    \item We develop the semiparametric efficiency bound for the ATE in our context (Section~\ref{sec:semiparametric_efficiency});
    \item Based on this lower bound, we derive efficient covariate density and propensity score (Section~\ref{sec:efficiency_prob});
    \item Using these efficient probabilities, we design an adaptive experiment for efficient ATE estimation (Section~\ref{sec:opt_experiment}).
\end{itemize}

\section{Problem Setting}
\label{sec:problem}
This section provides our problem formulation. We define potential outcomes and observations separately by following the Neyman–Rubin causal model \citep{Neyman1923,rubin1974}. Then, we discuss estimating the ATE from the observations gathered from an adaptive experiment.

\subsection{Potential Outcomes}
There is a binary treatment $a \in\{1, 0\}$. Let us define the corresponding potential outcome by $Y(a)$. 
Let $X\in \mathcal{X}\subset \mathbb{R}^d$ be a $d$-dimensional covariate, where $\mathcal{X}$ is the space. 

For any $x\in\mathcal{X}$, let $P_0(x)$ be the true conditional distribution of $(Y_1, Y_0)$ given $X = x$ whose density is given as $r^1(y(1), y(0)\mid x)$.
%\footnote{Here, $Y(1)$ and $Y(0)$ are not necessarily independent. In analysis, we use $\mathbbm{1}[A_t = 1]Y_t(1)\mathbbm{1}[A_t = 0]Y_t(0) = 0$.} 
Let $\mathbb{E}[\cdot \mid X = x]$ and $\mathrm{Var}\big(\cdot \mid X = x\big)$ be expectation and variance operators over $P_0(x)$. 
Let us denote the true mean and variance of $Y(a)$ conditioned by $X = x \in\mathcal{X}$ by $\mu_0(a)(x) = \mathbb{E}[Y(a)\mid X=x]$ and
$\sigma^2_0(a)(x) = \mathrm{Var}(Y(a) \mid X=x)$, respectively. 

We cannot observe both $Y(1)$ and $Y(0)$ simultaneously. We only observe $Y(A)$ for a treatment-assignment indicator $A \in \{1, 0\}$. We define how we observe data in Section~\ref{sec:aae}.  

We make the following regularity assumption.
\begin{assumption}
\label{asm:dist}
There exist known universal constants $\underline{C}$ and $\overline{C}$ such that $0 < \underline{C} < \overline{C} < \infty$ and for any $x\in\mathcal{X}$, under $P_0(x)$, $\big|\mu_0(a)(x)\big| < \overline{C}$ and $\underline{C} < \sigma^2_0(a)(x) < \overline{C}$ hold. 
\end{assumption}

\subsection{Average Treatment Effect}
As a causal effect, we are interested in the difference $Y(1) -Y(0)$, referred to as an ITE in some cases. However, because one of $Y(1)$ and $Y(0)$ is unobservable owing to the counterfactual property, we cannot observe it directly. Therefore, we focus on the expected value of $Y(1) -Y(0)$. 

Let $q(x)$ be the density of an evaluation covariate distribution. Our interest is in the ATE over $q(x)$, defined as
\begin{align*}
    \theta_0 &\coloneqq \mathbb{E}_{X\sim q(x),Y(a) \sim r^a(y\mid X)}[Y(1) - Y(0)],
\end{align*}
where $\mathbb{E}_{X\sim q(x),Y(a) \sim r^a(y\mid X)}$ is expectation over $X$ generated from $q(x)$ and $Y(a)$ given $X$ generated from $r^a(y\mid X)$. Note that $\theta_0 = \mathbb{E}_{X\sim q(x)}[\mathbb{E}[Y(1) - Y(0)\mid X]] = \mathbb{E}_{X\sim q(x)}\left[\mu_0(1)(X) - \mu_0(0)(X)\right]$ holds, where $\mathbb{E}_{X\sim q(x)}$ denotes expectation over $X$ generated from $q(x)$.

Our interest is to estimate the ATE with a smaller asymptotic variance via the following adaptive experiment.

For simplicity, $q(x)$ is assumed to be known to the experimenter. This assumption can be mitigated by considering different settings. In Appendix~\ref{sec:rejection}, we introduce another setting where we can observe covariates $X_t$ generated from $q(x)$ and decide whether we use it by the rejection sampling. 

\begin{assumption}
\label{asm:covariate}
    The covariate density $q(x)$ is known.
\end{assumption}

\subsection{Active Adaptive Experiments}
\label{sec:aae}
Consider an adaptive experiment with $T$ rounds, $[T] \coloneqq\{1,2,\dots, T\}$. In the experiment, in each $t\in[T]$, an experimenter decides a covariate density $p_t(x)$ and observes an experimental unit with covariates $X_t \sim p_t(x)$. Given the observed covariates $X_t$, the experimenter decides propensity score $w_t(a\mid X_t)$ and assigns treatment $A_t \in \{1, 0\}$, where $A_t \in \{1, 0\}$ is a \emph{treatment indicator}. Then, the experimenter observes an outcome $Y_t=\mathbbm{1}[A_t = 1]Y_{t}(1) + \mathbbm{1}[A_t = 0]Y_{t}(0)$, where $(Y_t(1), Y_t(0))$ is an i.i.d. copy of $(Y(1), Y(0))$. After round $T$, the experimenter estimates $\theta_0$ by using the observations $\{(Y_{t}(1), Y_{t}(0)\}^T_{t=1}$.

In our adaptive experiment, an experimenter can optimize the probabilities of $p_t(x)$ and $w_t(a\mid X_t)$ to obtain an estimator of $\theta_0$ with a smaller asymptotic variance.

In summary, the experimenter takes the following steps:
\begin{description}[topsep=0pt, itemsep=0pt, partopsep=0pt, leftmargin=*]
\item[Step~1:] based on past observations, an experimenter decides a covariate density $p_t$ and a propensity score $w_t$;
\item[Step~2:] the experimenter observes covariate $X_t\sim p_t(x)$;
\item[Step~3:] based on past information, the experimenter decides a propensity score $w_t(a\mid X_t)$;
\item[Step~4:] the experimenter assigns treatment $A_t = a$ with probability $w_t(a\mid X_t)$;
\item[Step~5:] the experimenter observes outcome $Y_t=\mathbbm{1}[A_t = 1]Y_{t}(1) + \mathbbm{1}[A_t = 0]Y_{t}(0)$, where $Y_{t}(a)\sim r^a(y\mid X_t)$;
\item[Step~6:] after iterating Steps~1--5 for $t=1,2,\dots, T$, in round $T$, we estimate the ATE.
\end{description}
We assume that $r^a(y\mid x)$ is invariant across rounds, and $Y_t(a)$ are drawn from the density independently of the other round's random variables; that is, $\{(Y_{t}(1), Y_{t}(0)\}^T_{t=1}$ is i.i.d., but $p_t(a)$ and $w_t(a\mid x)$ can take different values across rounds based on past observations.

Thus, the adaptive experiments yields observations $\{(Y_t, A_t, X_t)\}^T_{t=1}$, where $(Y_t, A_t, X_t)$ and $(Y_s, A_s, X_s)$ are correlated over time for $t \neq s$, that is, the samples are not i.i.d. Let $\mathcal{F}_{t-1}=\{X_{t-1}, A_{t-1}, Y_{t-1}, \dots, X_{1}, A_1, Y_{1}\}$ be the history. The probability $w_t(a\mid x)$ is a function of a covariate $X_t$, an action $A_t$, and history $\mathcal{F}_{t-1}$.

We refer to an experiment with this setting as an \emph{active adaptive experiment with covariate choice}. 

\begin{remark}[Stable unit treatment value assumption (SUTVA)]
\label{rem:sutva}
The DGP implies SUTVA; that is, $r^a(y\mid x)$ is invariant against $A_t$ across $t\in[T]$ \citep{Angrist1996}.
\end{remark}

\begin{figure}[t]
  \centering
    \includegraphics[width=80mm]{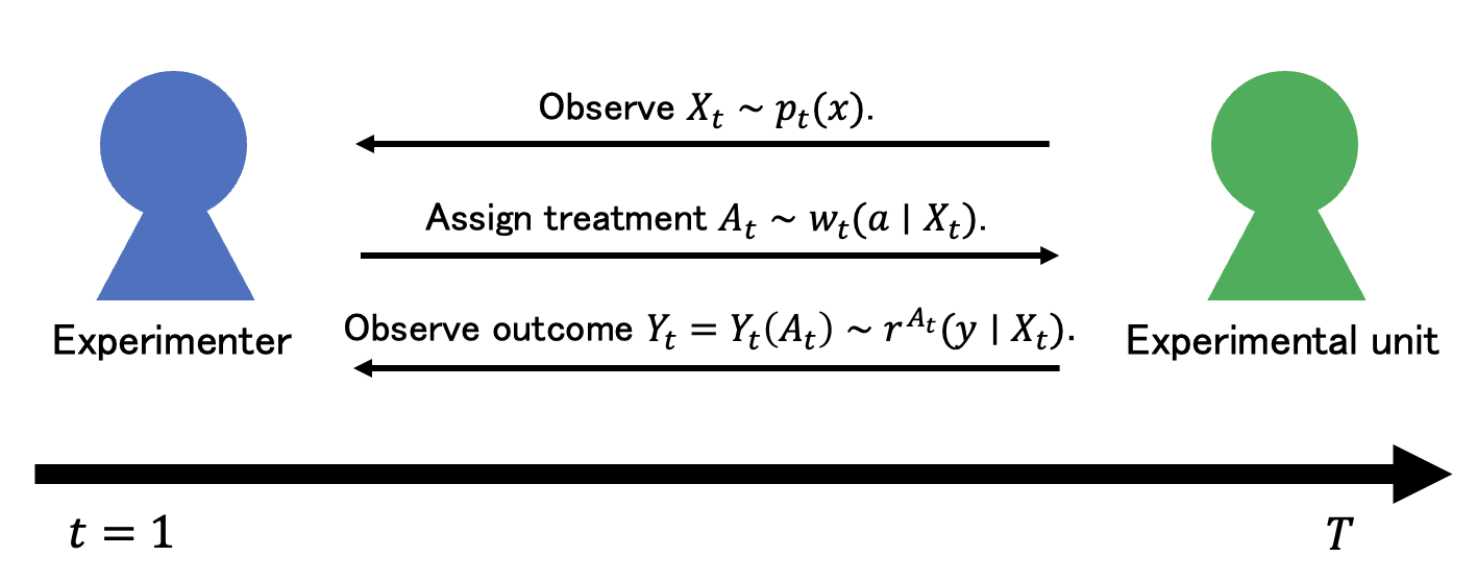}
\vspace{-5mm}
\caption{Active adaptive experiment.}
\label{fig:exp3}
\vspace{-5mm}
\end{figure}

\section{Semiparametric Efficiency Bound}
\label{sec:semiparametric_efficiency}
This section provides a lower bound for the asymptotic variance of ATE estimators. We focus on semiparametric efficiency bound for a class of regular estimators. The semiparametric efficiency bound generalize the Cramer-Rao lower bound, where the latter considers parametric models, while the former considers non- and semiparametric models. 

The semiparametric efficiency bound is a lower bound for a class of regular estimators, which are $\sqrt{T}$-consistent estimators for all DGPs and remain so under perturbations of size $1/\sqrt{T}$ to the DGP
(see \citet{Vaart1998}).

For simplicity, suppose that $(Y_t, A_t, X_t)$ is an i.i.d. copy of $(Y, A, X)$. 
Let $r^a(y\mid x)$ be a density of $Y(a)$ conditioned on $X = x$, $w(d\mid x)$ be a probability of $A$ conditioned on $X = x$, and $p(x)$ be a density of $X$. 
Suppose that $(Y, A, X)$ is generated as $(Y, A, X) \sim f(y, a, x)$, where
\begin{align*}
    f(y, a, x) = \prod_{d\in\{1, 0\}}\left(r^d(y(d)\mid x) w(d\mid x)\right)^{a=d}p(x),
\end{align*}
where $r^d(y(d)\mid x)$ is the marginal distribution of $Y(d)$. 
Note that in our setting, we can manipulate $w(a\mid x)$ and $p(x)$ and update them during an experiment.

Furthermore, we \emph{hypothetically} assume access to a set of covariates sampled from $q(x)$. Then, to derive the semiparametric efficiency bound, we consider a stratified sampling scheme, where two independent datasets $\mathcal{D}_T \coloneqq \{(Y_t, A_t, X_t)\}^T_{t=1}$ and $\widetilde{\mathcal{D}}_S \coloneqq \{\widetilde{X}_s\}^S_{s=1}$ are available. Note that $(Y_t, A_t, X_t) \iid \prod_{d\in\{1, 0\}}\Big\{r^d(y\mid x) w(d\mid x)\Big\}^{a = d}p(x)$ and $\widetilde{X}_s \iid q(x)$. First, we derive the semiparametric efficiency bound under a situation where $\mathcal{D}_T$ and $\widetilde{\mathcal{D}}_S$ are available. Then, since we know $q(x)$ and can generate infinitely many samples from $q(x)$, we focus on the semiparametric efficiency bound at a limit of $S\to\infty$.

Under this DGP and $q(x)$,
we obtain the following theorem. The proof is shown in Appendix~\ref{appdx:thm:lower}.
\begin{theorem}
\label{thm:lower}
Under Assumption~\ref{asm:covariate}, given $\mathcal{D}_T$ and $\widetilde{\mathcal{D}}_S$, as $S\to\infty$, the semiparametric efficiency bound for $\theta_0$ with fully nonparametric models is
\begin{align*}
&\tau(w, p)=\mathbb{E}_{X\sim q(x)}\left[\left(\frac{\sigma^2_0(1)(X)}{w(1\mid X)} + \frac{\sigma^2_0(0)(X)}{w(0\mid X)}\right)\frac{q(X)}{p(X)}\right].
\end{align*}
\end{theorem}

Theorem~\ref{thm:lower} is a lower bound; that is, a theoretically best asymptotic variance. It means that even though we minimize the lower bound, it is still unclear whether an estimator whose asymptotic variance aligns with the lower bound exists. In Section~\ref{sec:opt_experiment}, we actually design an adaptive experiment under which an ATE estimator's asymptotic variance aligns with the lower bound minimized over $w$ and $p$. 

This semiparametric efficiency bound can be regarded as a functional of $p$ and $w$. During our adaptive experiments, we can minimize the bound regarding $p$ and $w$. 
%Note that the semiparametric efficiency gives a lower bound under fixed $p$ and $w$.

\section{Efficient Probabilities}
\label{sec:efficiency_prob}
This section investigates $w$ and $p$ that minimizes $\tau(w, p)$ by considering the following functional optimization problem:
\begin{align}
\label{eq:functional}
    (w^*, p^*) := \argmin_{(w, p)\in\mathcal{W}\times \mathcal{P}}\tau(w, p), 
\end{align}
where $\mathcal{W}$ is a set of all measurable functions
$w:\mathcal{X}\to \Delta \coloneqq \big\{ u = (u_1\ u_2\ \dots\ u_K)^\top \in [0, 1]^K \mid \sum^K_{k=1}u_k = 1 \big\}$ with its $a$-th element $w(a\mid x)$, and $\mathcal{P}$ is a set of all measurable functions $p:\mathcal{X} \to \mathbb{R}^+$ such that $\int_{x\in\mathcal{X}} p(x) \mathrm{d}x = 1$. We refer to the minimizers (stationary point) $p^*$  and $w^*$ as \emph{efficient} covariate density and propensity score since the semiparametric efficiency bound is minimized under them. 

\subsection{Efficient Propensity Score}
\label{sec:eff_prop}
First, we find an efficient propensity score that minimizes the semiparametric efficiency bound given some $p(x)$. Given $p$, a stationary point regarding $w$ is given as the point-wise optimizer for each $x \in \mathcal{X}$. That is, given $x$, for $w(x) \coloneqq \{w(a\mid x)\}_{a\in[K]} \in (0, 1)^K$, we have $w^*(x) = \argmax_{w\in\Delta} \Big\{\frac{\sigma^2_0(1)(x)}{w(1\mid x)} + \frac{\sigma^2_0(0)(x)}{w(0\mid x)}\Big\}\frac{q^2(x)}{p(x)}$. Then, 
for each $x\in\mathcal{X}$, the efficient propensity score is 
\begin{align*}
&w^*(a\mid x) = \frac{\sigma_0(a)(x)}{\sigma_0(1)(x) + \sigma_0(0)(x)}\qquad \forall a \in \{1, 0\}.
\end{align*}
This efficient propensity score is referred to as the \emph{Neyman allocation}, which has been traditionally employed in designing experiments with binary treatments without covariate choice \citep{Neyman1934OnTT,Laan2008TheCA,Hahn2011,Kato2020adaptive,adusumilli2022minimax,Dai2023}. 

\subsection{Efficient Covariate Density}
Next, we find an efficient covariate density that minimizes the semiparametric efficiency bound for fixed $w$.

\begin{theorem} 
\label{thm:covariate_prob}
For each $x\in\mathcal{X}$, the minimizer $p^*$ of the semiparametric efficiency bound $\tau(w, p)$ with a fixed $w$ is 
\begin{align*}
&p^*(x) = \frac{\sqrt{\frac{\sigma^2_0(1)(x)}{w(1\mid x)} + \frac{\sigma^2_0(0)(x)}{w(0\mid x)}}}{\mathbb{E}_{X\sim q(x)}\left[\sqrt{\frac{\sigma^2_0(1)(X)}{w(1\mid X)} + \frac{\sigma^2_0(0)(X)}{w(0\mid X)}}\right]}q(x).
\end{align*}
\end{theorem}
By substituting $w^*$ obtained in Section~\ref{sec:eff_prop}, we have
\[p^*(x) = \frac{\sigma_0(1)(x) + \sigma_0(0)(x)}{\mathbb{E}_{X\sim q(x)}\left[\sigma_0(1)(x) + \sigma_0(0)(x)\right]}q(x).\]

This result implies that we can reduce the efficiency bound by drawing samples whose sums of conditional standard deviations of potential outcomes, $\sigma_0(1)(x) + \sigma_0(0)(x)$, are large. That is, we assign treatments to experimental units with high-variance outcomes more than those with less-variance outcomes to accurately estimate the ATE. 

For example, in clinical trials, to reduce the ATE, we gather elderly with high blood pressure more than healthy youth because the former may have more volatile outcomes. We illustrate this intuition in Figure~\ref{fig:concept_figure2}.

\begin{figure}[t]
  \centering
    \includegraphics[width=80mm]{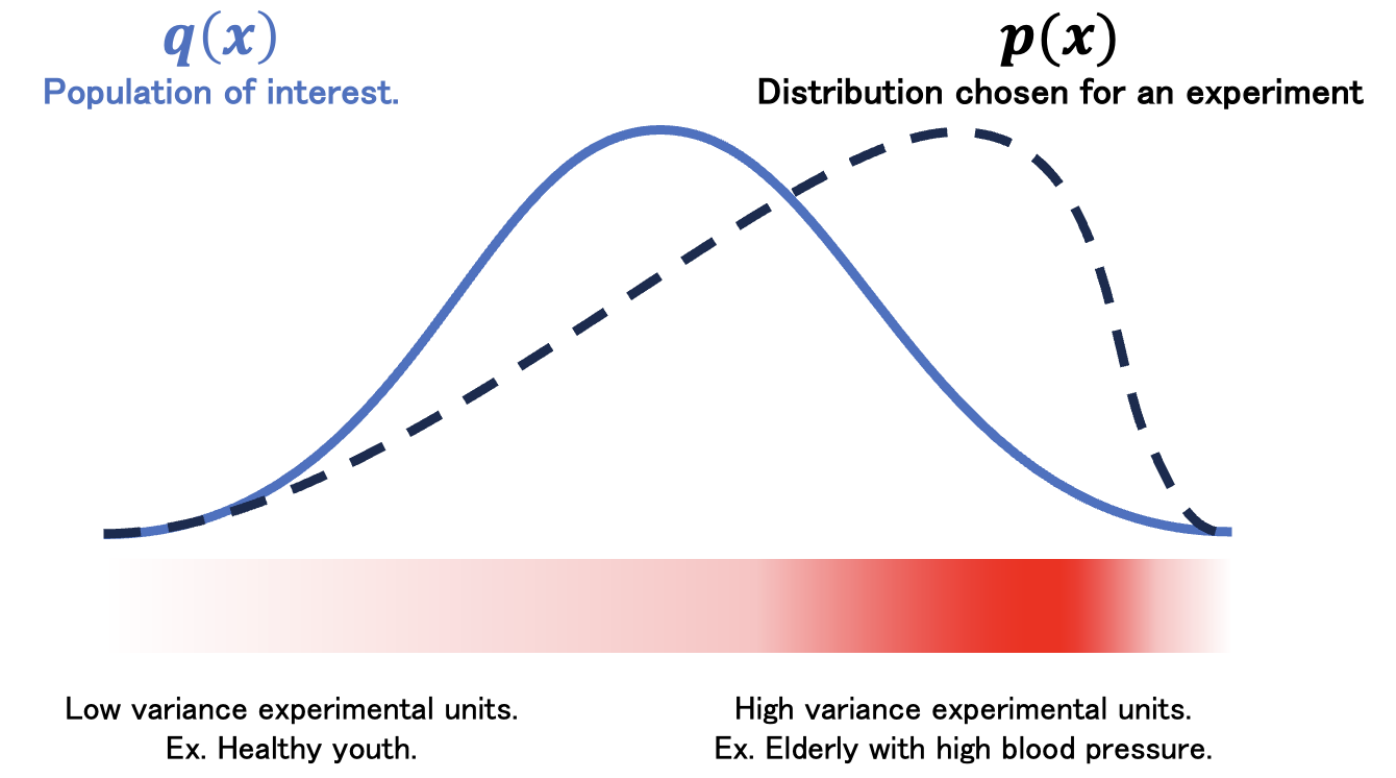}
\vspace{-5mm}
\caption{Efficient covariate density. The distribution with the blue line represents the covariate density of interest, while that with the dotted black line represents the covariate density in an experiment.}
\label{fig:concept_figure2}
\vspace{-3mm}
\end{figure}

\subsection{Minimized Semiparametric Efficiency Bound}
The minimizer of $\tau(w, p)$ is given as $w^*$ and $p^*$. By substituting them, we can obtain the minimized semiparametric efficiency bound $\tau^* \coloneqq \tau(w^*, p^*)$ as follows:
\begin{align*}
\tau^* &= 
    \left(\mathbb{E}_{X \sim q(x)}\left[\sqrt{\frac{\sigma^2_0(1)(X)}{w^*(1\mid X)} + \frac{\sigma^2_0(0)(X)}{w^*(0\mid X)}}\right]\right)^2\\
    &= 
    \Big(\mathbb{E}_{X \sim q(x)}\left[\sigma_0(1)(X) + \sigma_0(0)(X)\right]\Big)^2.
\end{align*}

\subsection{Efficiency Gain}
This section investigates how the efficient probability assignment probability and covariate density improve the semiparametric efficiency bound (efficiency gain).

\paragraph{Efficient propensity score} As a result of the optimization, for each $x\in\mathcal{X}$, for any $w \in \mathcal{W}$, we have
\begin{align*}
    &\frac{\sigma^2_0(1)(x)}{w(1\mid x)} + \frac{\sigma^2_0(0)(x)}{w(0\mid x)}\\
    &\geq \frac{\sigma^2_0(1)(x)}{w^*(1\mid x)} + \frac{\sigma^2_0(0)(x)}{w^*(0\mid x)}  = \big(\sigma_0(1)(x) + \sigma_0(0)(x)\big)^2. 
\end{align*}
We confirm how the lower bound is minimized by considering the uniform sampling $w^{\mathrm{Unif}}(a\mid x) = 1/2$ for each $a\in\{1, 0\}$ and $x\in\mathcal{X}$. We have 
\begin{align*}
    &\frac{\sigma^2_0(1)(x)}{w^{\mathrm{Unif}}(1\mid x)} + \frac{\sigma^2_0(0)(x)}{w^{\mathrm{Unif}}(0\mid x)}= 2\Big(\sigma^2_0(1)(x) + \sigma^2_0(0)(x)\Big)\\
    &\geq \big(\sigma_0(1)(x) + \sigma_0(0)(x)\big)^2= \frac{\sigma^2_0(1)(x)}{w^*(1\mid x)} + \frac{\sigma^2_0(0)(x)}{w^*(0\mid x)}.
\end{align*}
The difference between them is $\big(\sigma_0(1)(x) - \sigma_0(0)(x)\big)^2$. Thus, as the difference between $\sigma_0(1)(x)$ and $\sigma_0(0)(x)$ becomes large, the efficiency gain also becomes large.

\paragraph{Efficient covariate density}
Let us fix the propensity score as $w^*(a\mid x) = \frac{\sigma_0(1)(x)}{\sigma_0(1)(x) + \sigma_0(0)(x)}$ and consider a case where the covariate density is fixed at $q(x)$. 
When $(Y, A, X)$ is generated from a density $\prod_{d\in[K]}\big\{r^a(y\mid x) w^*(d\mid x)\big\}^{d=a}q(x)$, 
the semiparametric efficiency bound becomes 
\begin{align*}
\widetilde{\tau} &\coloneqq \mathbb{E}_{X\sim q(x)}\Bigg[\frac{\sigma^2_0(1)(X)}{w^*(1\mid X)} + \frac{\sigma^2_0(0)(X)}{w^*(0\mid X)}\Bigg]\\
&= \mathbb{E}_{X\sim q(x)}\left[\big(\sigma_0(1)(X) + \sigma_0(0)(X)\big)^2\right].
\end{align*}
That is, any regular estimators cannot obtain a smaller asymptotic variance than the lower bound when samples are generated from $\prod_{d\in[K]}\Big\{r^a(y\mid x) w^*(d\mid x)\Big\}^{d=a}q(x)$.

Next, we compare the efficiency bound $\tau^*$ when samples are generated from $\prod_{d\in[K]}\Big\{r^a(y\mid x) w^*(d\mid x)\Big\}^{d=a}p^*(x)$. 
From the Jensen inequality, it follows that
\begin{align*}
    \tau^* &= \Big(\mathbb{E}_{X\sim q(x)}\left[\sigma_0(1)(X) + \sigma_0(0)(X)\right]\Big)^2\\
    &\leq 
\mathbb{E}_{X\sim q(x)}\left[\big(\sigma_0(1)(X) + \sigma_0(0)(X)\big)^2\right]= \widetilde{\tau}.
\end{align*}
The equality holds when $\big(\sigma_0(1)(x) + \sigma_0(0)(x)\big)^2$ is the same across $x \in \mathcal{X}$. Therefore, when $\big(\sigma_0(1)(x) + \sigma_0(0)(x)\big)^2$ varies across $x$, we can reduce the semiparametric efficiency bound by optimizing $p(x)$; that is $\tau^* < \widetilde{\tau}$.

\section{Active Adaptive Experimental Design for Treatment Effect Estimation}
\label{sec:opt_experiment}
We design an adaptive experiment with covariate choice, referred to as the \emph{Active-Adaptive-Sampling-and-Augmented-Inverse-Probability-Weighting} (AAS-AIPWIS) experiment because our approach resembles active learning, which adaptively optimizes covariate density and the propensity score, and then estimate the ATE by the AIPWIW estimator. We first define our experiment in Section~\ref{sec:aaed} and then show the asymptotic property in Section~\ref{sec:asymp_dist}.

\subsection{AAS-AIPWIS Experiment}
\label{sec:aaed}
This section provides the definition of our experiment, which consists of the \emph{sampling phase} and \emph{estimation phase}. In round $t \in [T]$, an experimenter engages in the sampling phase, where the experimenter optimizes the covariate density and the propensity score using past observations up to round $t$; then, the experimenter observes covariates and assigns a treatment following the efficient probabilities. After round $T$, the estimation phase starts, where we estimate the ATE using $\{(Y_t, A_t, X_t)\}^T_{t=1}$. We describe the details below, and its pseudo-code is provided in Algorithm~\ref{alg}.

\paragraph{Sampling phase}
Define the following estimators:
\begin{itemize}[topsep=0pt, itemsep=0pt, partopsep=0pt, leftmargin=*]
    \item \textbf{Conditional mean estimator.} Let $\widehat{\mu}_t(a)(x)$ be an $\mathcal{F}_{t-1}$-measurable estimator of $\mu_0(a)(x)$ for each $a\in\{1, 0\}$ and $x\in\mathcal{X}$ such that $|\widehat{\mu}_t(a)(x)| < \overline{C}$ holds, and it converges to $\mu_0(a)(x)$ almost surely.
    \item \textbf{Conditional second-moment estimator.} Let $\widehat{\nu}_t(a)(x)$ be an $\mathcal{F}_{t-1}$-measurable estimator of $\nu_0(a)(x) \coloneqq \mathbb{E}[Y^2_t(a)\mid x]$ for each $a\in\{1, 0\}$ and $x\in\mathcal{X}$. 
    \item \textbf{Conditional variance estimator.} Let $\widehat{\sigma}^2_t(a)(x)$ be an $\mathcal{F}_{t-1}$-measurable bounded estimator of $\sigma^2_0(a)(x)$ for each $a\in\{1, 0\}$ and $x\in\mathcal{X}$, defined as
    \begin{align*}
        \widehat{\sigma}^2_t(a)(x) = \mathrm{thre}\left(\widehat{\nu}_t(a)(x) - \widehat{\mu}^2_t(a)(x); \underline{C}, \overline{C}\right),
    \end{align*}
    where $\mathrm{thre}(A; a, b) = A\ \mathrm{if}\ a \leq A \leq b,\ \ \mathrm{thre}(A; a, b) = a\ \mathrm{if}\ A < a,\ \mathrm{and}\ \mathrm{thre}(A; a, b) = b\ \mathrm{if}\ A > b$.
\end{itemize}
Specifically, we construct $\widehat{\mu}_t(a)(x)$ and $\widehat{\nu}_t(a)(x)$ only by using past observations $\mathcal{F}_{t-1} = \{(Y_s, A_s, X_s)\}^{t-1}_{s=1}$. 

For $\widehat{\mu}_t(a)(x)$ and $\widehat{\nu}_t(a)(x)$, we can use nonparametric estimators, such as the nearest neighbor regression estimator and kernel regression estimator, which have been proven to converge to the true function almost surely under a bounded sampling probability $\widehat{w}_t$ by \citet{yang2002} and \citet{qian2016kernel}. In our simulation studies, we employ the Nadaraya-Watson kernel regression \citep{Nadaraya1964,Watson1964}, whose consistency has been proved by \citet{qian2016kernel} in an adaptive experiment when the propensity score is lower bounded by a universal constant. 

Then, we estimate the efficient propensity score as 
\begin{align}
\label{eq:est_efficient_prob}
&\widehat{p}_t(x) \coloneqq \frac{\widehat{\sigma}_t(1)(x) + \widehat{\sigma}_t(0)(x)}{\mathbb{E}_{X\sim q(x)}\left[\widehat{\sigma}_t(1)(x) + \widehat{\sigma}_t(0)(x)\right]}q(x)\nonumber\\
&\widehat{w}_t(a\mid x) \coloneqq \frac{\widehat{\sigma}_t(a)(x)}{\widehat{\sigma}_t(1)(x) + \widehat{\sigma}_t(0)(x)},\\
&\mathrm{where}\ \ \widehat{\theta}_t(x) \coloneqq \widehat{\mu}_t(1)(x) - \widehat{\mu}_t(0)(x).\nonumber
\end{align}

We assumed that we can directly compute $\mathbb{E}_{X\sim q(x)}\left[\widehat{\sigma}_t(1)(X) + \widehat{\sigma}_t(0)(X)\right]$ by using the known $q(x)$. For example, we generate enough samples from $q(x)$ and compute the expectation. 

Instead of the true expectation, by using past observations, we can approximate it as $\mathbb{E}_{q(x)}\left[\widehat{\sigma}_t(1)(X) + \widehat{\sigma}_t(0)(X)\right] \approx$
\begin{align*}
&\frac{1}{t - 1}\sum^{t-1}_{s=1}\big(\widehat{\sigma}_t(1)(X_s) + \widehat{\sigma}_t(0)(X_s)\big) \frac{q(X_s)}{\widehat{p}_s(X_s)}.
\end{align*}

\paragraph{Estimation phase}
In the estimation phase, we estimate the ATE $\theta_0$ by using the AIPWIW estimator, defined as
\begin{align}
\label{eq:AIPWIW}
    \widehat{\theta}_T \coloneqq \frac{1}{T}\sum^T_{t=1}\Psi_t(Y_t, A_t, X_t; \widehat{w}_t, \widehat{p}_t),
\end{align}
where $\Psi_t(Y_t, A_t, X_t; \widehat{w}_t, \widehat{p}_t) \coloneqq$
\begin{align*}
    &\Bigg(\frac{\mathbbm{1}[A_t = 1]\big(Y_t(1) - \widehat{\mu}_t(1)(X_t)\big)}{\widehat{w}_t(1\mid X_t)} \\
    &\ \ \ \ \ \ \ \ \ \ \ \ - \frac{\mathbbm{1}[A_t = 0]\big(Y_t(0) - \widehat{\mu}_t(0)(X_t)\big)}{\widehat{w}_t(0\mid X_t)}\Bigg)\frac{q(X_t)}{\widehat{p}_t(X_t)}\\
    \nonumber
    &\ \ \ \ \ \ \ \ \ \ \ \ \ \ \ \ \ \ \ \ \ \ \ \ \ \ \ \ \ \ \ \ \ \ \ \ \ \ \ \ \ \ \ \ \ \ \ \ \ \ \ + \mathbb{E}_{\widetilde{X}\sim q(x)}\left[\widehat{\theta}_t(\widetilde{X})\right],
\end{align*}
where $\mathbb{E}_{\widetilde{X}\sim q(x)}\left[\widehat{\theta}_t(\widetilde{X})\right]$ is computed using the known $q(x)$. This estimator is a variant of the \emph{Adaptive AIPW (A2IPW) estimator} proposed in \citet{Kato2020adaptive} with importance weighting \citep{Uehara2020}.

This estimator removes a sample selection bias caused by the propensity score, which assigns treatments based on covariates. Furthermore, by employing conditional mean estimators $\widehat{\mu}_t$ constructed from past observations, we can employ $\mathbb{E}\left[\Psi_t(Y_t, A_t, X_t; \widehat{w}_t, \widehat{p}_t)\mid \mathcal{F}{t-1}\right] = \theta_0$. This method is a form of sample-splitting technique, as discussed in \citep{Laan2008TheCA}, frequently utilized in semiparametric analysis and recently refined as DML \citep{ChernozhukovVictor2018Dmlf}.

\subsection{Asymptotic Distribution and Efficiency}
\label{sec:asymp_dist}
This section provides an asymptotic distribution of $\widehat{\theta}_T$. First, we make the following assumptions.
\begin{assumption}
\label{asm:subgaussian}
    Random variables $Y_t(a) - \mu_a(a)(x)$ and $X_t - \mathbb{E}[X_t]$ are mean-zero sub-Gaussian.
\end{assumption}
\begin{assumption}
\label{asm:nuisance_consistency}
For all $a\in\{1, 0\}$ and any $x\in\mathcal{X}$, 
    \begin{align*}
        \widehat{\mu}_t(a)(x) \xrightarrow{\mathrm{a.s.}} \mu_0(a)(x),\quad \widehat{\sigma}^2_t(a)(x) \xrightarrow{\mathrm{a.s.}} \sigma^2_0(a)(x)
    \end{align*}
    hold as $t\to \infty$.
\end{assumption}
This assumption also implies $\widehat{p}_t(x) \xrightarrow{\mathrm{a.s.}} p^*(x)$. 
\begin{theorem}
\label{thm:asymp_dist}
Consider the AAS-AIPWIS experiment. 
    Suppose that Assumptions~\ref{asm:dist}--\ref{asm:covariate} and \ref{asm:subgaussian}--\ref{asm:nuisance_consistency} hold. Then, 
    \begin{align*}
        \sqrt{T}\left(\widehat{\theta}_T - \theta_0\right) \xrightarrow{\mathrm{d}}\mathcal{N}(0, \tau^*)
    \end{align*}
    holds as $T \to \infty$. 
\end{theorem}
Thus, the asymptotic variance of $\widehat{\theta}_T$ aligns with the semiparametric efficiency bound minimized concerning the covariate density $p(x)$ and the propensity score $w(a \mid x)$. This result implies that our proposed experiment successfully reduces the asymptotic variance, and we cannot achieve a smaller variance within a class of regular estimators.

In the proof, we use the sample-splitting technique developed by \citet{klaassen1987}, \citet{Laan2008TheCA}, \citet{Luedtke2016}, \citet{ChernozhukovVictor2018Dmlf}, and \citet{Kato2021adr}. With this technique, we can derive the asymptotic distribution only with the convergence rate condition in Assumption~\ref{asm:nuisance_consistency}, without assuming the Donsker condition.

We also note that $\widehat{\theta}_T$ is consistent for $\theta_0$, even without assuming Assumption~\ref{asm:subgaussian}.
\begin{corollary}
    Suppose that Assumptions~\ref{asm:dist}--\ref{asm:covariate} and \ref{asm:subgaussian} hold. Then, $
        \widehat{\theta}_T - \theta_0$ holds as $T \to \infty$. 
\end{corollary}
This property is based on the double robustness of the AIPW estimator $\widehat{\theta}_T$, which guarantees the consistency of $\widehat{\theta}_T$ if either $\mu_0(a)(x)$ or $w^*(a \mid x)$ is consistently estimated. In our case, since we know the true value of $w^*(a \mid x)$, the consistency is guaranteed by the double robustness. For this property, even if we cannot estimate $\mu_0(a)(x)$ consistently, we can still consistently estimate $\theta_0$.

Note that the asymptotic variances in the lower bound and our estimator in Theorem~\ref{thm:asymp_dist} do not depend on the dimension $d$. The dimensionality may affect whether the nuisance estimators $\widehat{\mu}_t$ and $\widehat{\sigma}^2_t$ satisfy Assumption~\ref{asm:nuisance_consistency}. For example, if $d$ is significantly large, we might use arguments about sparsity or benign overfitting to guarantee that Assumption~\ref{asm:nuisance_consistency} holds. Once Assumption~\ref{asm:nuisance_consistency} holds, the asymptotic variance of $\widehat{\theta}_T$ does not depend on $d$.

\begin{algorithm}[tb]
   \caption{AAS-AIPWIS experiment}
   \label{alg}
\begin{algorithmic}
   \STATE {\bfseries Parameter:} Positive constants $C_{\mu}$ and $C_{\sigma^2}$.
   \STATE {\bfseries Initialization:} 
   \FOR{$t\in\{1, 2\}$}
   \STATE Assign $A_t=1$ at $t = 1$ and $A_t = 0$ at $t = 2$. For each $a\in\{1, 0\}$, set $\widehat{w}_{t}(a\mid x) = 1/2$.
   \ENDFOR
   \FOR{$t=K+1$ to $T$}
   \STATE Construct $\widehat{p}_t(x)$ following \eqref{eq:est_efficient_prob}. 
   \STATE Draw $X_t$ from $\widehat{p}_t(x)$. 
   \STATE Construct $\widehat{w}_{t}(1\mid X_t)$ by using the variance estimators.
   \STATE Draw $\xi_t$ from the uniform distribution on $[0,1]$. 
   \STATE $A_t = 1$ if $\xi_t \leq \widehat{w}_{t}(1\mid X_t)$ and $A_t = 0$ otherwise. 
   \ENDFOR
   \STATE Estimate $\theta_0$ by using the AIPWIW estimator $\widehat{\theta}_T$ in \eqref{eq:AIPWIW}.
\end{algorithmic}
\end{algorithm}

\section{Simulation Studies}
This section presents simulation studies. There are two objectives for the simulation. First, through simulation studies, we clarify the problem setting. Second, for the problem, we confirm the soundness of our designed experiment.

\begin{figure}[t]
  \centering
    \includegraphics[width=70mm]{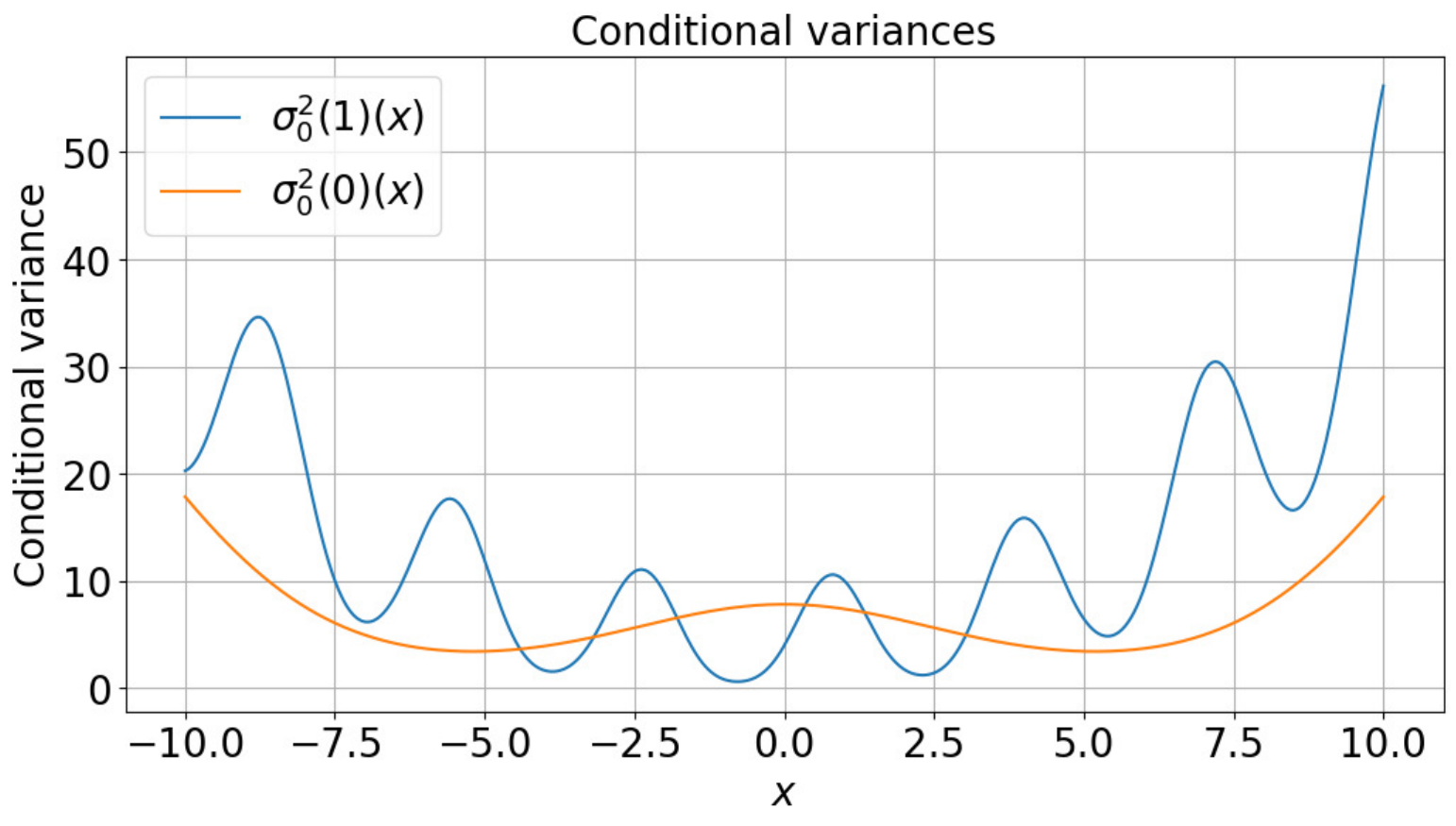}
\vspace{-5mm}
\caption{Heterogeneous variances.}
\label{fig:hetero}
\vspace{1mm}
    \includegraphics[width=70mm]{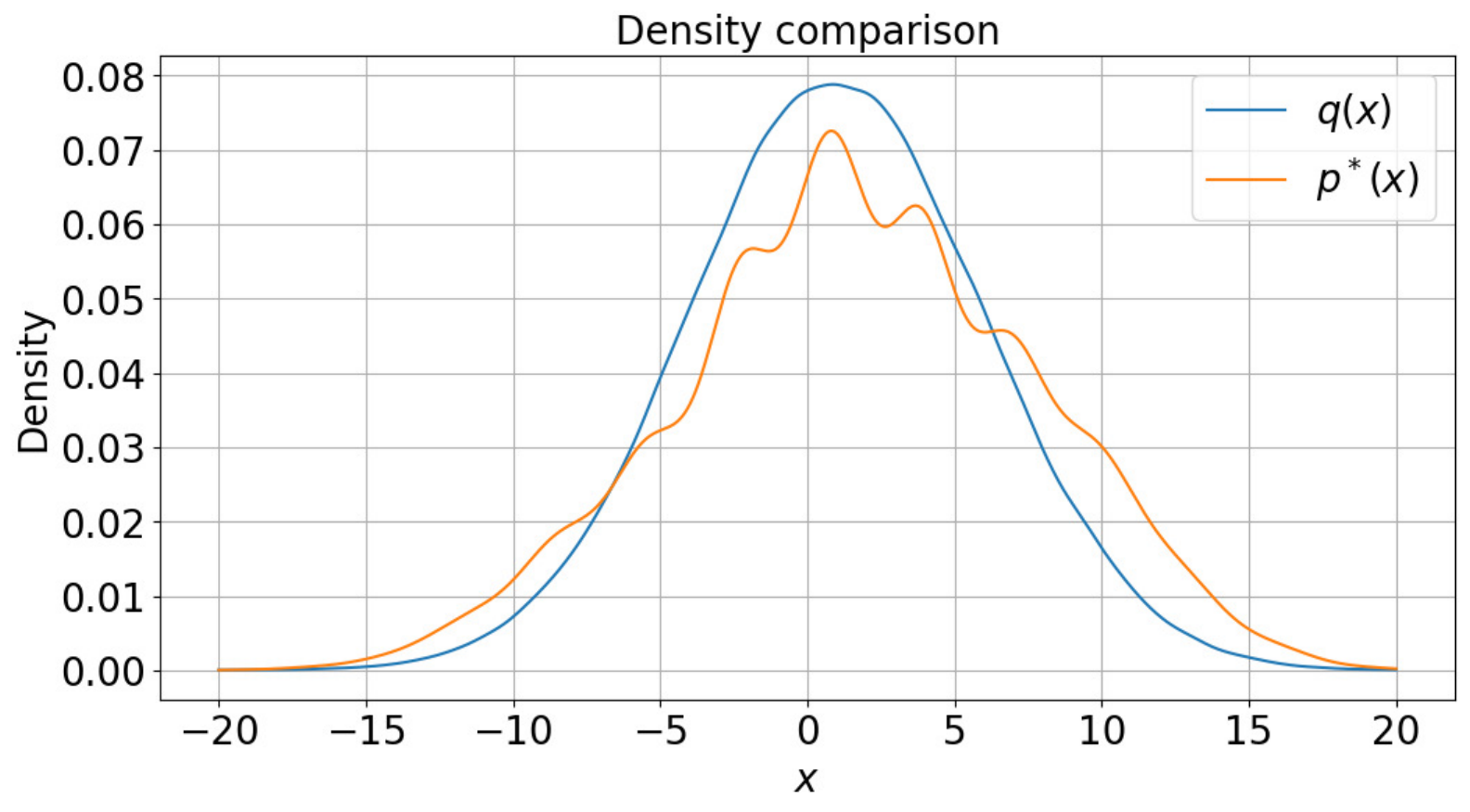}
\vspace{-5mm}
\caption{Comparison between $q(x)$ and $p^*(x)$ in the experiment with the covariates following the Gaussian distribution.}
\label{fig:gaussian}

\vspace{1mm}

    \includegraphics[width=70mm]{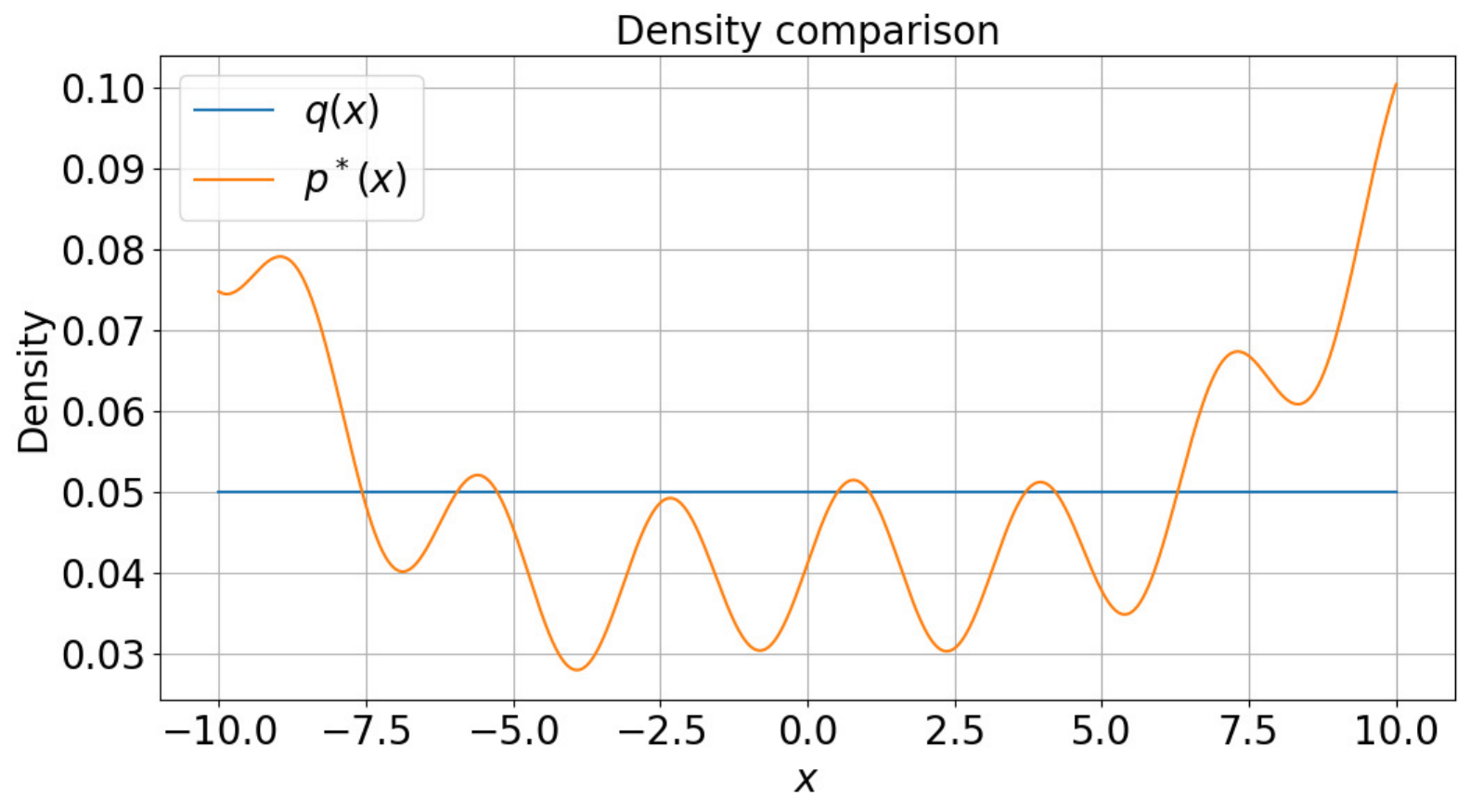}
\vspace{-5mm}
\caption{Comparison between $q(x)$ and $p^*(x)$ in the experiment with the covariates following the Uniform distribution.}
\vspace{1mm}
\label{fig:uniform}
\vspace{-5mm}
\end{figure}

\begin{figure}[t]
  \centering
    \includegraphics[width=70mm]{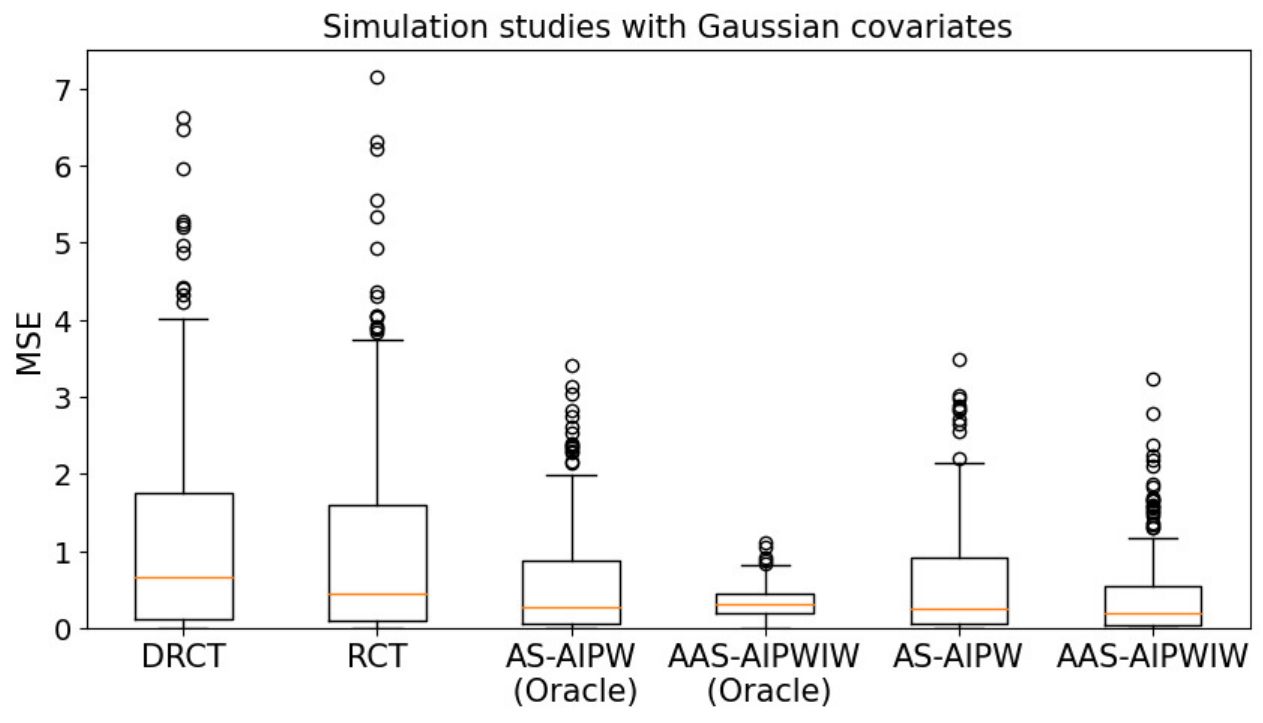}
\vspace{-5mm}
\caption{Results of simulation studies with covariates following with a Gaussian distribution $\mathcal{N}(1, 25)$.}
\label{fig:gaussian_res}
\vspace{1mm}
\includegraphics[width=70mm]{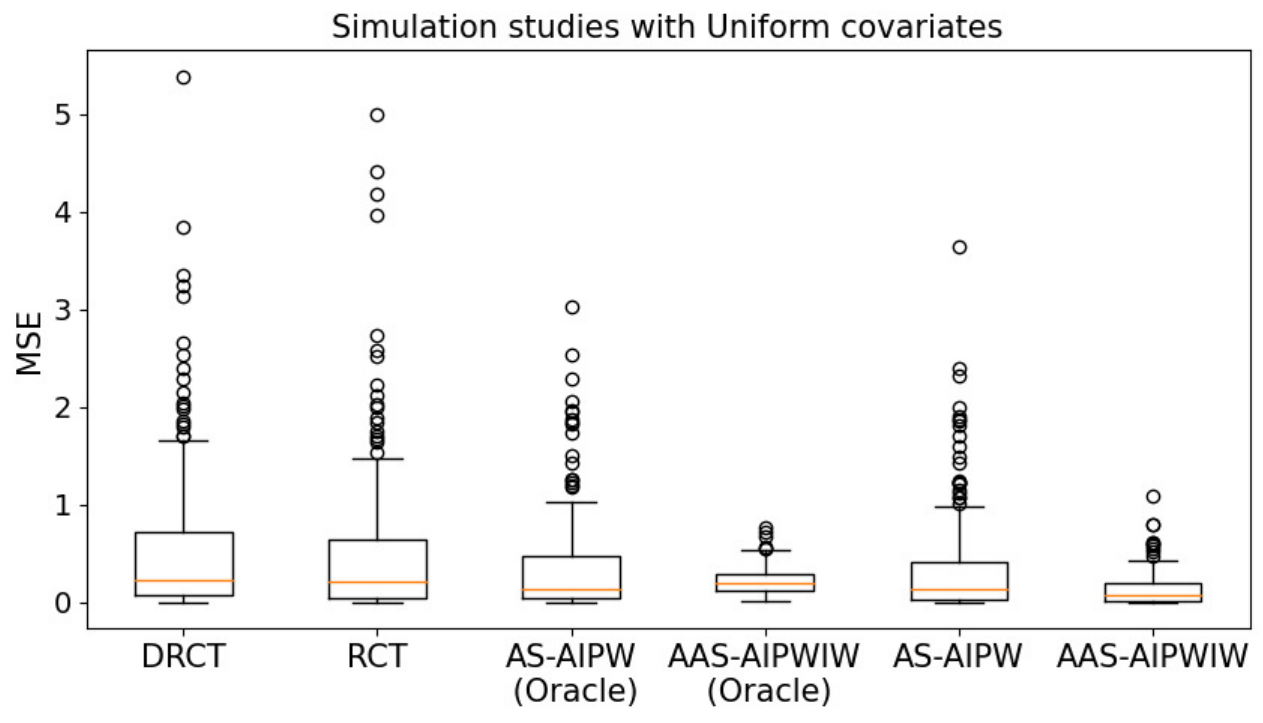}
\vspace{-5mm}
\caption{Results of simulation studies with covariates following with a Uniform distribution $\mathcal{U}(-10, 10)$.}
\label{fig:uniform_res}
\vspace{-5mm}
\end{figure}

\paragraph{Basic setup}
Let $T = 2000$. We specify the means as $\mathbb{E}_{X\sim q(x)}[Y(1)] = 10$ and $\mathbb{E}_{X\sim q(x)}[Y(0)] = 7$, under which the ATE is $\theta_0 = 3$. Let covariates $x$ be one-dimensional. For any $x$, the conditional means are 
\begin{align*}
    &\mu_0(1)(x) \coloneqq  C_1\left(- x + 3x^2 - 1\right),\\
    &\mu_0(0)(x) \coloneqq C_0\left(0.1x + 0.2\right),
\end{align*}
where $C_1$ and $C_0$ are parameters so that $\mathbb{E}_{X\sim q(x)}[\mu_0(1)(X)] = \mathbb{E}_{X\sim q(x)}[Y(1)] = 10$ and $\mathbb{E}_{X\sim q(x)}[\mu_0(0)(X)] = \mathbb{E}_{X\sim q(x)}[Y(0)] = 7$. 

\paragraph{Heterogeneous variances}
We define $\sigma^2_0(a)(x)$  as follows:
\begin{align*}
    &\sigma^2_0(1)(x) \coloneqq 2 + 1.2\mathrm{sin}(2x) +(x + x^2)/25\\
    &\sigma^2_0(0)(x) \coloneqq 2 + 0.8\mathrm{cos}(x / 2)  +x^2 / 50.
\end{align*}
For $x\in[-10, 10]$, we plot $\sigma^2_0(a)(x)$ in Figure~\ref{fig:hetero}. 

\paragraph{Alternative experiments}
We compare our designed experiment with the deterministic RCT (DRCT), the RCT, and the Adaptive-Sampling-and-AIPW (AS-AIPW) experiment. In the deterministic RCT, we deterministically assign treatment $1$ if $t$ is odd and assign treatment $0$ if $t$ is even. In the RCT, we randomly assign treatments with probability $w(1\mid x) = w(0\mid x) = 0.5$ for any $x$. We estimate the ATE by using the sample average in the DRCT and the RCT. In the AS-AIPW experiment, we optimize the propensity score but do not optimize the covariate density (experimental units are just generated from $q(x)$); then, we estimate the ATE with the AIPW estimator without importance weighting. The AS-AIPW estimator is basically the same method proposed in \citet{Laan2008TheCA}, \citet{Hahn2011}, and \citet{Kato2020adaptive}, though we simplified them for simplicity. 

Additionally, we conduct the AS-AIPW and AAS-AIPWIW experiments under the true $w^*$ and $p^*$; that is, they are known at the beginning of an experiment. We refer to the experiments as the AS-AIPW (Oracle) and the AAS-AIPWIW (Oracle). Although we use the true values for $w^*$ and $p^*$ in the sampling phase, we estimate $\mu(a)(x)$ in the estimation phase to construct the AIPW estimator.

\textbf{Performance measure.}
In each of our simulations, we conduct $200$ experiments and compute the performance measures. In this section, in the y-axis of Figures~\ref{fig:gaussian_res}--\ref{fig:uniform_res}, we report the empirical MSEs for $\theta_0$; that is, $\frac{1}{200}\sum^{200}_{i=1}(\theta_0 - \hat{\theta}^{(i)}_T)^2$, where $\hat{\theta}^{(i)}_T$ denotes the estimator of ATE in the $i$-th trial.  Note that the asymptotic MSEs correspond to the asymptotic variances. 

\textbf{Covariates following with a Gaussian distribution.}
Let $q(x)$ be the density of the Gaussian distribution $\mathcal{N}(1, 25)$ with mean $1$ and variance $25$. For $q(x)$, we plot $q(x)$ and $p^*(x)$ in Figure~\ref{fig:gaussian}. We show the results in Figure~\ref{fig:gaussian_res}. 

\paragraph{Covariates following with a Uniform distribution}
Let $q(x)$ be a density of the Uniform distribution with support $[-10, 10]$. For $q(x)$, we plot $q(x)$ and $p^*(x)$ in Figure~\ref{fig:uniform}. We show the results of simulation studies in Figure~\ref{fig:uniform_res}. 

\paragraph{Results} 
From Figures~\ref{fig:gaussian_res} and \ref{fig:uniform_res}, we find that our experiments successfully reduce the MSE compared to the others. Surprisingly, the AS-AIPW and the AAS-AIPW show almost identical or better performances than the oracle ones. Such a paradox often occurs in semiparametric analysis, as pointed out by \citep{Henmi2004paradox,Hitomi_Nishiyama_Okui_2008}.

\section{Related Work and Discussion}
In this section, we discuss related topics.

\subsection{Related Work}
We provide a literature review of related work in the field of adaptive experimental design for efficiently estimating ATEs.

Adaptive sampling has long been investigated with various purposes \citep{Wald1945,Robbins1952,Solomon1970, Rosenberger2001,bai2024primer}. Pioneering contributions for ATE estimation by \citet{Laan2008TheCA} and \citet{Hahn2011} lay the foundation for this field. 

Building upon \citet{Laan2008TheCA} and \citet{Hahn2011}, \citet{Kato2020adaptive} and \citet{Meehan2022} propose refined methods by extending the existing methodologies. These studies derive the semiparametric efficiency bound for ATE estimators and specify the efficient propensity score that minimizes this bound. 

This approach is further refined by \citet{cook2023semiparametric}, who proves the asymptotic normality of the A2IPW estimator under weaker assumptions than those of \citet{Kato2020adaptive}. They also propose an adaptive experiment that ensures finite-sample stability by truncating the propensity score \citep{Waudby-Smith2024}. Furthermore, they derive tighter concentration inequalities using results from \citet{Howard2020TimeuniformNN}, \citet{Waudby-Smith2023}, and \citet{Waudby-Smith2023_nonpara}. \citet{gupta2021efficient} designs an adaptive experiment involving instrumental variables. There are other studies discussing related topics \citep{Blackwell2022, Westort2021}.

Recent work by \citet{Dai2023} introduces an approach based on online stochastic projected gradient descent for experimental design when the optimal treatment allocation rule depends on unknown parameters. Their study rigorously addresses the issue of variance estimation, previously considered negligible at the order of $O(1/\sqrt{T})$ under asymptotic approximation. They propose a method to effectively estimate variance and conduct experimental design while accounting for the impact of variance estimation errors. \citet{zhao2023adaptive} also investigates the variance estimation problem in Neyman allocation.

The AIPW estimator is popular in adaptive experiments. Utilizing its variance reduction and martingale properties, the A2IPW estimator is introduced in \citet{Kato2020adaptive} based on the arguments in \citet{Laan2008TheCA} and \citet{ChernozhukovVictor2018Dmlf}. \citet{Kato2021adr} and \citet{li2023double} propose DML for adaptive experiments. \citet{hadad2019} and \citet{Zhan2021} also focus on this topic, presenting different estimators under a different setting. There are also several approaches for statistical inference in adaptive experiments \citep{Zhang2020, Bibaut2021}.

In the context of active learning, \citet{Sugiyama2006} demonstrates that an appropriate covariate density choice can reduce the asymptotic variance of regression coefficients. This setting is extended to the pool-based setting by \citet{Sugiyama2009}, which is relevant to our context. \citet{Wu2019pool} further extends these approaches to adaptive experiments, contrasting with earlier non-adaptive experiments where observations are gathered without parameter estimation during the experiment. Independently, \citet{zrnic2024active} explores active learning for causal inference.

\citet{Uehara2020} considers off-policy evaluation, a generalization of ATE estimation, under a covariate shift. They show that ATE estimation with importance weighting, as proposed by \citet{Shimodaira2000}, with nonparametrically estimated density-ratio, may not yield an estimator whose asymptotic variance aligns with the semiparametric efficiency bound.

Best-Arm Identification (BAI), another related topic, is a type of stochastic multi-armed bandit (MAB) problem \citep{Thompson1933, Lai1985}, aimed at identifying the best treatment (arm) \citep{Audibert2010, Kaufman2016complexity, Russo2016}. BAI has two settings: fixed-confidence \citep{Garivier2016} and fixed-budget. The fixed-confidence setting is related to sequential testing \citep{bechhofer1968sequential, Garivier2016}, while the fixed-budget setting is BAI with a fixed sample size \citep{Bubeck2009, Bubeck2011, degenne2023existence}. \citet{kato2024locally, kato2024worstcase, kato2024agna} investigate and generalize the Neyman allocation in fixed-budget BAI, proving its asymptotic optimality.

In the MAB problem, covariates are referred to as contexts. While various algorithms with contexts have been investigated in the MAB problem, it is still under exploration how to utilize contexts in BAI. In the fixed-confidence setting, \citet{russac2021abn}, \citet{kato2021role}, and \citet{simchi2024experimentation} investigate the BAI with covariates, which is referred to as contextual information. In the fixed-budget setting, \citet{kato2024adaptivepolicylearning} designs an adaptive experiment for policy learning.

\subsection{Sample Size Calculation}
In statistical hypothesis testing, sample size calculation is important. Here, we discuss how to decide $T$ in our method. Let the null and alternative hypotheses be $H_0: \theta_0 = \mu$ and $H_1(\Delta): \theta_1 = \mu + \Delta$, where $\mu \in \mathbb{R}$ and $\Delta > 0$ are some constants. Let $\alpha \in (0, 1)$ be a significance level such that $\mathbb{P}(\mathrm{reject}\ H_0) = \alpha$, and $\beta \in (0, 1)$ be the Type II error level. Let $R^{\alpha}_T$ be a rejection region defined based on $\alpha$. Then, we are interested in deciding the minimum sample size for hypothesis testing, which is defined as $T^*_{\alpha,\beta}(\Delta) \coloneqq \min\big\{T: \mathbb{P}_{H_1(\Delta)}\big(\widehat{\theta}_T \in R^{\alpha}_T\big) \geq 1 - \beta \big\}$ for some estimator $\widehat{\theta}_T$ of $\theta_0$. Let us assume that $\sqrt{T}\big(\widehat{\theta}_T - \mu\big) \xrightarrow{\mathrm{d}} \mathcal{N}(0, V)$ under $H_0$ and $\sqrt{T}\big(\widehat{\theta}_T - \mu - \Delta\big) \xrightarrow{\mathrm{d}} \mathcal{N}(0, V)$ under $H_1$. Then, when using $t$-statistics for hypothesis testing, we have $T^*_{\alpha,\beta}(\Delta) = \frac{V}{\Delta^2}\big(z_{1-\alpha/2} - z_{\beta}\big)$, where $z_{c}$ is the $c$-quantile of the standard normal distribution \citep{Balsubramani2016,Kato2020adaptive}. For example, in the AAS-AIPWIW, we substitute $\tau^*$ into $V$. This result implies that we can reduce the sample size $T$ for hypothesis testing by reducing the asymptotic variance. For the details of sample size calculation, see Appendix~G in \citet{Kato2020adaptive}.

\subsection{Extensions}
We can consider several extensions of our designed experiment. In Appendix~\ref{appdx:extension}, we propose three extensions: (i) rejection sampling, (ii) batch design, and (iii) experiments with multiple treatments. In the rejection sampling setting, we do not assume that $q(x)$ is known, but we sequentially observe an experimental unit generated from the unknown $q(x)$. By employing rejection sampling, we can decide whether to accept the observed experimental unit and if we accept it, we assign a treatment to it. This setting is feasible without knowing $q(x)$. In (ii), we can update $\widehat{p}_t$ and $\widehat{w}_t$ at certain rounds. For example, we can update them at $t = 100, 200, 300, \dots$ and need to keep using the same $\widehat{p}_t$ and $\widehat{w}_t$ between the two update timings. In (iii), we consider there are treatments more than or equal to three.

In this study, we assumed that $q(x)$ is known. For example, if the population we are interested in consists of the people of a certain country, we can obtain information about their distribution through government statistics. Then, in an experiment, we consider gathering people with different distributions using the method we propose. In applications, we can also consider estimating $q(x)$ using observations from a pilot experiment. If there are $M$ samples from $q(x)$, it is possible to estimate $q(x)$ using those samples. In this case, if $M$ diverges independently of $T$, the current results can be directly applicable (asymptotics of $M\to \infty$ and $T\to \infty$). However, if the divergence rate of $M$ depends on $T$, we need to slightly modify our results and methods for each setting of applications.

\section{Conclusion}
In this study, we designed an adaptive experiment to estimate ATEs efficiently. In our adaptive experiment, we can optimize covariate densities and propensity scores in each round based on past observations obtained until that round. To derive the efficient covariate density and propensity score, we focused on the semiparametric efficiency bound for the ATE. Based on the derived efficient probabilities, we designed an adaptive experiment using the AIPW estimator with importance weighting. Under the designed experiment, we demonstrated that the asymptotic variance of the AIPW estimator aligns with the efficiency bound, minimized with respect to the covariate density and propensity score.

\section*{Impact Statement}
This paper presents an innovative approach to machine learning, specifically focusing on the efficient estimation of ATEs through adaptive experimental design. Our work has the potential to significantly impact various sectors, including healthcare, economics, and social sciences. However, if not applied carefully, the optimization techniques could lead to biased sampling or unethical manipulation of experimental conditions, potentially harming subjects or leading to incorrect conclusions. Furthermore, the adaptive nature of the experiments, relying on sequential sampling and treatment assignment based on past observations, may raise data privacy and consent issues, especially in sensitive fields like healthcare. In conclusion, while our work aims to advance the field of machine learning by providing a novel framework for efficient ATE estimation through adaptive experiments, it is crucial to consider the ethical implications and societal impacts of its application. Researchers and practitioners should be mindful of the potential risks and strive to apply these methods in a responsible and equitable manner.

\bibliographystyle{icml2024}
\bibliography{AEE.bbl}

\appendix
\onecolumn

\section{Proof of Theorem~\ref{thm:lower}}
\label{appdx:thm:lower}
In this section, we prove Theorem~\ref{thm:lower}, the semiparametric efficiency bound for regular estimators of $\theta_0$. Our derivation of the semiparametric efficiency bound follows the approaches of \citet{hahn1998role}, \citet{hirano2003}, \citet{narita2019counterfactual}, and \citet{Uehara2020}. 

Their proof considers a nonparametric model for the distribution of potential outcomes and defines regular subparametric models. Then, (i) we characterize the tangent set for all regular parametric submodels, (ii) verify that the parameter of interest is pathwise differentiable, (iii) verify that a guessed semiparametric 
efficient influence function lies in the tangent set, and (iv) calculate the expected square of the influence function.

\subsection{Stratified Sampling}
As discussed in Section~\ref{sec:semiparametric_efficiency}, we consider a stratified sampling scheme to derive the semiparametric efficiency bound. This sampling scheme assumes the availability of the following two datasets:
\begin{align*}
    &\mathcal{D}_T \coloneqq \{(Y_t, A_t, X_t)\}^T_{t=1},\\
    &\widetilde{\mathcal{D}}_S \coloneqq \{\widetilde{X}_s\}^S_{s=1},
\end{align*}
where $(Y_t, A_t, X_t) \iid \prod_{d\in\{1, 0\}}\Big\{r^d(y\mid x) w(d\mid x)\Big\}^{d=a}p(x)$ and $\widetilde{X}_s \iid q(x)$. We refer to $\mathcal{D}_T$ as the observation data and $\widetilde{\mathcal{D}}_S$ as the evaluation data.
Because we know $q(x)$, we can obtain infinite samples from $q(x)$. Therefore, we first derive the semiparametric efficiency bound under the sampling scheme and then consider $S \to \infty$ and then consider the semiparametric efficiency bound regarding the asymptotics of $T$. 

Our derivation of the semiparametric efficiency bound is inspired by the technique used in \citet{narita2019counterfactual} and \citet{Uehara2020}, but we modify it since we need to consider the asymptotics of $T$ and $S$ separately.

\subsection{Parametric Submodels}
\label{appdx:submodel}
We consider regular parametric submodels with a parameter $\beta \in \mathbb{R}$ to discuss the semiparametric efficiency bound under the stratified sampling scheme. 

We denote the probability densities under the parametric submodels by 
\begin{align}
\label{eq:param_sub}
    f(y, a, x; \beta) = \prod_{d\in\{1, 0\}}\left(r^d(y\mid x; \beta) w(d\mid x; \beta)\right)^{a = d}p(x; \beta)
\end{align}
and $q(\widetilde{x}; \beta)$, where $r^d(y\mid x; \beta)$ is the conditional density of $Y(d)$ given $X = x$, $w(d\mid x; \beta) = \mathrm{Pr}(A_t = d\mid X = x; \beta)$, $p(x;\beta)$ is the density of $X$ in observation data, and  $q(\widetilde{x};\beta)$ is the density of $\widetilde{X}$ in evaluation data. 

We also denote $\prod_{d\in\{1, 0\}}\left\{r^d(y\mid x; \beta) w(d\mid x; \beta)\right\}^{a = d}$ by $f(y, a\mid x; \beta)$.

Then, the log densities under the parametric submodels are given as
\begin{align*}
    &\log \big(f(y, a, x; \beta)\big)\\
    &= \mathbbm{1}[a = 1]\Big(\log \big(r^1(y\mid x)\big) + \log \big(w(1\mid x; \beta)\big)\Big) + \mathbbm{1}[a = 0] \Big(\log \big(r^0(y\mid x)\big) + \log\big(w(0\mid x; \beta)\big)\Big) + \log \big(p(x; \beta)\big),
\end{align*}
and $\log\big(q(\widetilde{x}; \beta)\big)$. 

We consider the joint log-likelihood of $\mathcal{D}_T$ and $\mathcal{D}_S$, which is defined as
\begin{align*}
    \sum^T_{t=1}\log \left(f(Y_t, A_t, X_t; \beta)\right) + \sum^S_{s=1}\log\big(q(\widetilde{X}_s; \beta)\big). 
\end{align*}

By taking the derivative of $\sum^T_{t=1}\log \left(f(Y_t, A_t, X_t; \beta)\right) + \sum^S_{s=1}\log\big(q(\widetilde{X}_s; \beta)\big)$ with respect to $\beta$, we can obtain the corresponding score as
\begin{align*}
    &S(W;\beta) \coloneqq \frac{\mathrm{d}}{\mathrm{d}\beta} \left(\sum^T_{t=1}\log \left(f(Y_t, A_t, X_t; \beta)\right) + \sum^S_{s=1}\log\big(q(\widetilde{X}_s; \beta)\big)\right)\\
    &= \sum^T_{t=1}\left(\mathbbm{1}[A_t=1] \left\{s^1(Y_t\mid X_t; \beta) + \frac{\dot{w}(1\mid X_t;\beta)}{w(1\mid X_t; \beta)}\right\} + \mathbbm{1}[A_t=0] \left(s^0(Y_t\mid X_t; \beta) + \frac{\dot{w}(0\mid X_t;\beta)}{w(0\mid X_t; \beta)}\right) + \zeta(X_t; \beta)\right)\\
    &\ \ \ \ \ \ \ + \sum^S_{s=1}\widetilde{\zeta}(\widetilde{X}_s; \beta),
\end{align*}
where $s^d$, $\dot{\omega}$, $\zeta$, and $\widetilde{\zeta}$ are defined as
\begin{align*}
    s^d(y\mid x; \beta) &\coloneqq \frac{\mathrm{d}}{\mathrm{d}\beta} \log r^d(y\mid x; \beta),\\
    \dot{\omega}(d\mid x; \beta) &\coloneqq \frac{\mathrm{d}}{\mathrm{d}\beta} w(d\mid x; \beta),\\
    \zeta(x; \beta) &\coloneqq \frac{\mathrm{d}}{\mathrm{d}\beta} \log p(x; \beta),\\
    \widetilde{\zeta}(\widetilde{x}; \beta) &\coloneqq \frac{\mathrm{d}}{\mathrm{d}\beta} \log q(\widetilde{x}; \beta).
\end{align*}

\subsection{ATE Function}
Next, we consider the ATE under the parametric submodel. Under our defined parametric submodel, let $\theta(\beta)$ be our parameter of interest as a function of $\beta$, which is given as
\begin{align*}
    \theta(\beta) \coloneqq \mathbb{E}_{X \sim q(x; \beta),(Y, A)\sim f(y, a\mid X; \beta)}\left[Y(1) - Y(0)\right].
\end{align*}
We have the following lemma for the derivative of $\theta(\beta)$.
\begin{lemma} 
Under the parametric submodel defined in \eqref{eq:param_sub}, we have
\label{lem:deriv}
    \begin{align*}
        &\frac{\partial \theta(\beta)}{\partial \beta} = \mathbb{E}_{W}\left[Y(1)s^1(Y(1)\mid X; \beta) - Y(0)s^0(Y(0)\mid X; \beta) + \theta(\widetilde{X}; \beta)\widetilde{\zeta}(\widetilde{X}; \beta)\right],
    \end{align*}
    where 
    \begin{align*}
            \theta(X; \beta) &\coloneqq \mathbb{E}_{Y(1)\sim r^1(y(1)\mid X; \beta)}[Y(1)\mid X] - \mathbb{E}_{Y(0)\sim r^0(y(0)\mid X; \beta)}[Y(0)\mid X].
    \end{align*}
\end{lemma}
The statement follows from Appendix~B in \citet{narita2019counterfactual}. 

\subsection{Cram\'{e}r-Rao Lower Bound}
\label{appdx:cramer}
To obtain an intuition about the semiparametric efficiency bound, 
we discuss the Cram\'{e}r-Rao lower bound under a specification of the parametric submodels, which simplifies an argument about the efficiency bound.

To consider the Cram\'{e}r-Rao lower bound, we specify the score functions as follows:
\begin{align}
\label{eq:specify}
    s^1(y\mid x; \beta) &= \frac{1}{T}\frac{\big( y - \mathbb{E}_{Y(1)\sim r^d(y(1)\mid x; \beta)}[Y(1)\mid X = x]\big)}{w(1\mid x)}\\
    s^0(y\mid x; \beta) &= - \frac{1}{T}\frac{\big( y - \mathbb{E}_{Y(0)\sim r^d(y(0)\mid x; \beta)}[Y(0)\mid X = x]\big)}{w(0\mid x)}\nonumber\\
    \dot{w}(d\mid X;\beta_0) &= 0,\nonumber\\
    \zeta(x; \beta) &= 0,\nonumber\\
    \widetilde{\zeta}(\widetilde{x}; \beta) &= \frac{1}{S}\Big\{\theta(\widetilde{x}; \beta) - \theta(\beta)\Big\}.\nonumber
\end{align}
We consider the Cram\'{e}r-Rao lower bound when these score functions are true ones in the DGP. 

We fix $T$ and $S$ and consider a case where $m$ i.i.d. copes of 
\[W\coloneqq \Big(X_1,\cdots,X_{T},A_1,\cdots,A_{T},Y_1,\cdots,Y_{T},\widetilde{X}_1,\cdots,\widetilde{X}_{S}\Big)\]
are available, 
where $W$ is generated from $\prod^T_{t=1}f(Y_t, A_t, X_t; \beta)\prod^S_{s=1}q(\widetilde{x}; \beta)$. 
Then, we consider the Cram\'{e}r-Rao lower bound for estimating $\theta_0$. 

Let $\mathbb{E}_W$ is an expectation operator over $(Y, A, X)$ and $\widetilde{X}$. 

Given a distribution whose density is given as $\prod^T_{t=1}f(Y_t, A_t, X_t; \beta)\prod^S_{s=1}q(\widetilde{X}_s; \beta)$, the Fisher information is given as
\[\mathbb{E}_W\left[\left(\nabla_{\beta}\left(\sum^T_{t=1}\log \left(f(Y_t, A_t, X_t; \beta)\right) + \sum^S_{s=1}\log \left(q\left(\widetilde{X}_s; \beta\right)\right)\right)\right)^2\right]^{-1}.\]

Then, the Cram\'{e}r-Rao lower bound for unbiased estimators of the function $\theta(\beta)$ under the parametric submodel is
\begin{align*}
    &\mathrm{CR}(\theta)\coloneqq \nabla_{\beta}\theta(\beta)\mathbb{E}_W\left[\left(\nabla_{\beta}\left(\sum^T_{t=1}\log \left(f(Y_t, A_t, X_t; \beta)\right) + \sum^S_{s=1}\log \left(q\left(\widetilde{X}_s; \beta\right)\right)\right)\right)^2\right]^{-1}\nabla_{\beta}\theta(\beta).
\end{align*}

\begin{theorem}\label{thm:tabular}
Under the specification of \eqref{eq:specify}, the Cram\'{e}r-Rao lower bound $T\mathrm{CR}(\theta)$ is $\tau^*(\omega, p)$ as $S \to \infty$.
\end{theorem}

\begin{proof}[Proof of Theorem~\ref{thm:tabular}]
We defined the ATE as
\begin{align*}
    &\theta(\beta) \coloneqq \int \Big(y(1) - y(0)\Big) r^1\big(y(1) \mid \widetilde{x}; \beta\big)r^0\big( y(0) \mid \widetilde{x}; \beta\big)q(\widetilde{x}; \beta)\mathrm{d}\mu\big(y(1), y(0), \widetilde{x}\big).
\end{align*}
where $\mu$ is a baseline measure such as Lebesgue or a counting measure. The Cram\'{e}r-Rao lower bound for parametric models defined in \eqref{eq:specify} is given as follows:
\begin{align*}
    T\mathrm{CR}\big(\theta\big) = TC D^{-1}C^{\top}, \label{eq:cracra}
\end{align*}
where
\begin{align*}
   C &\coloneqq \mathbb{E}_{W}\Big[Y(1)s^1(Y(1)\mid X; \beta_{0}) - Y(0)s^0(Y(0)\mid X; \beta_{0}) +  \theta(\widetilde{X}; \beta_{0})\widetilde{\zeta}(\widetilde{X}; \beta_{0})\Big],  \nonumber\\
   D &\coloneqq \mathbb{E}_{W}\left[\left(\nabla_{\beta}\left(\sum^T_{t=1}\log \big(f(Y_t, A_t, X_t; \beta_0)\big) + \sum^S_{s=1}\log \big(q(\widetilde{X}_s; \beta_0)\right)\right)^2\right].\nonumber
\end{align*}
Here, we used Lemma~\ref{lem:deriv}. 

Under the specification in \eqref{eq:specify}, we have
\begin{align*}
    &T\mathrm{CR}\big(\theta\big)\\
    &= T\mathbb{E}_W\left[\frac{1}{T}\left(Y(1)\frac{\big(Y(1) - \mathbb{E}[Y(1)\mid X]\big)}{w(1\mid X)} + Y(0)\frac{\big(Y(0) - \mathbb{E}[Y(0)\mid X]\big)}{w(0\mid X)}\right) +  \frac{1}{S}\theta(\widetilde{X}; \beta_{0})\left(\theta(\widetilde{X}; \beta_{0}) - \theta(\beta_{0})\right)\right]\\
    &\ \ \ \times \mathbb{E}_W\left[T\left(\nabla_{\beta}\log \big(f(Y, A, X; \beta_0)\big)\right)^2 + S\left(\nabla_{\beta}\log \big(q(\widetilde{X}; \beta_0)\right)^2\right]^{-1}\\
    &= \mathbb{E}_W\left[\left(\frac{\big(Y(1) - \mathbb{E}[Y(1)\mid X]\big)^2}{w(1\mid X)} + \frac{\big(Y(0) - \mathbb{E}[Y(0)\mid X]\big)^2}{w(0\mid X)}\right) + \frac{T}{S}\left( \theta^2(\widetilde{X}; \beta_{0})\right)\right].
\end{align*}

Letting $S\to \infty$, we have 
\begin{align*}
    T\mathrm{CR}\big(\theta\big) = \mathbb{E}_W\left[\left(\frac{\big(Y(1) - \mathbb{E}[Y(1)\mid X]\big)^2}{w(1\mid X)} + \frac{\big(Y(0) - \mathbb{E}[Y(0)\mid X]\big)^2}{w(0\mid X)}\right) + \frac{T}{S}\left( \theta^2(\widetilde{X}; \beta_{0})\right)\right]= \tau^*(\omega, p).
\end{align*}
\end{proof}

\subsection{Derivation of the Semiparametric Efficiency Bound}
We extend the Cram\'{e}r-Rao lower for a case with semiparametric models.
As we did in Appendix~\ref{appdx:cramer}, we regard the whole $T + S$ data at hand as one observation and consider the case where we observe $m$ observations. Then, as $m$ goes to infinity, the total sample size $T':=Tm$ goes to infinity. Since we consider a case where $T + S$ samples are repeatedly observed from the same distribution independently, we can employ the conventional approach for deriving the semiparametric efficiency bound. Therefore, to employ this strategy, only in this section, we regard $T'$ as the total ``hypothetical'' sample size in the asymptotic analysis regarding $m\to \infty$ while fixing $T$. 

We explain the definition of the semiparametric efficient influence function. This is a function for one observation \[O=\Big(X_1,\cdots,X_{T},A_1,\cdots,A_{T},Y_1,\cdots,Y_{T},\widetilde{X}_1,\cdots,\widetilde{X}_{S}\Big).\]

The tangent set for this parametric submodel at $\beta = \beta_0$ is given as
\begin{align*}
&\mathcal{T} \coloneqq \left\{\sum^T_{t=1}\left(\sum_{d\in\{1, 0\}} \mathbbm{1}[A_t=d] \left(s^d(Y_t\mid X_t; \beta_0) + \frac{\dot{w}(d\mid X_t;\beta_0)}{w(d\mid X_t; \beta_0)}\right) + \zeta(X_t; \beta_0) \right) + \sum^S_{s=1}\widetilde{\zeta}\left(\widetilde{X}_s; \beta_0\right)  \in L_2(W)\right\},
\end{align*}
where we have
\begin{align*}
    &\int s^d(y\mid x; \beta_0) r^d(y\mid x; \beta_0) \mathrm{d}y = 0\qquad \forall x \in\mathcal{X},\\
    &\dot{w}(1\mid x;\beta_0) + \dot{w}(0\mid x;\beta_0) = 0\qquad \forall x \in\mathcal{X},\\
    &\int \zeta(x; \beta_0) p(x; \beta_0)\mathrm{d}x = 0,\\
    &\int \widetilde{\zeta}(\widetilde{x}; \beta_0) p(\widetilde{x}; \beta_0)\mathrm{d}\widetilde{x} = 0.
\end{align*}

Here, $\theta(\beta)$ is pathwise differentiable at $\beta_0$ if there exists a function $\psi(W) \in \mathcal{T}$ such that $\mathbb{E}[\psi^\top(W)\psi(W)] < \infty$ and for all regular parametric submodels
\begin{align}
\label{eq:riesz}
    \frac{\partial \theta(\beta_0)}{\partial \beta} \coloneqq \mathbb{E}_{W}\Big[\psi(W)S(W; \beta_0)\Big].
\end{align}

We refer to $\psi(W)$ as the semiparametric efficient influence function. As the following proposition, the semiparametric efficient influence function gives a lower bound for the asymptotic variance of regular estimators. 

\begin{proposition}[From Theorem~25.20 in \citet{Vaart1998}.]
The semiparametric efficient influence function $\psi(W)$ is the gradient of $\theta$  w.r.t the model $\mathcal{M}$, which has the smallest $L_2$-norm. It satisfies that for any regular estimator $\widehat{\theta}$ of $\theta_0$ regarding a given parametric submodel, $\mathrm{AMSE}\big(\widehat{\theta}\big)\geq \mathrm{Var}\big(\psi(W)\big)$, where $\mathrm{AMSE}\big(\widehat{\theta}\big)$ is the second moment of the limiting distribution of $\sqrt{m}\left(\widehat{\theta} - \theta_0\right)$. 
\end{proposition}

This states that $\mathrm{Var}\big(\psi(W)\big)$ is a lower bound in estimating $\theta_0$. We call $T\mathrm{Var}\big(\psi(W)\big)$ the semiparametric efficiency bound because what we want to consider is a lower bound of $\sqrt{T}\left(\widehat{\theta}-\theta_0\right)$. 

Lastly, we derive the semiparametric efficient influence function in the following lemma.
\begin{lemma}
\label{lem:basic_bound}
Under the parametric submodel defined in \eqref{eq:param_sub}, we have
\begin{align*}
    &\psi(W) = \frac{1}{T}\sum^T_{t=1}\left(\frac{\mathbbm{1}[A_t = 1]\big( Y_t - \mathbb{E}[Y_t(1)\mid X_t]\big)}{w(1\mid X_t)} - \frac{\mathbbm{1}[A_t = 0]\big( Y_t - \mathbb{E}[Y_t(0)\mid X]\big)}{w(0\mid X_t)}\right)\frac{q(X_t)}{p(X_t)} + \frac{1}{S}\sum^S_{s=1}\left(\theta_0(\widetilde{X}_s) - \theta_0\right).
\end{align*}
The (scaled) efficiency bound $T\mathrm{Var}\big(\psi(W)\big)$ is 
\[\mathbb{E}_{X\sim p(x)}\left[\left(\frac{\sigma^2_0(1)(X)}{w(1\mid X)} + \frac{\sigma^2_0(0)(X)}{w(0\mid X)}\right)\frac{q^2(X)}{p^2(X)}\right]+\frac{T}{S}\mathrm{Var}_{\widetilde{X}\sim q(x)}\big(\theta(\widetilde{X})\big).\]
\end{lemma}
From this lemma, as $S\to \infty$, we obtain $\tau^*$ in Theorem~\ref{thm:lower}.

\subsection{Proof of \cref{lem:basic_bound}}

Following existing studies, such as \citet{hahn1998role}, \citet{narita2019counterfactual}, and \citet{Uehara2020}, we consider a nonparametric model for the distribution of potential outcomes and define regular subparametric models. Then, we calculate a gradient (a candidate of the semiparametric efficient influence function) of the target function $\theta_0$, characterize the tangent set for all regular parametric submodels, verify that the parameter of interest is pathwise differentiable, verify that the
efficient influence function lies in the tangent space and calculates the expected square of the influence function.

As discussed in Appendices~\ref{appdx:submodel}, we define parametric submodels as
\begin{align*}
  p(W;\beta)&=\prod^T_{t=1}\prod_{d\in\{1, 0\}}\left(r^d(Y_t\mid X_t; \beta) w(d\mid X_t; \beta)\right)^{A_t = d}p(X_t; \beta)\prod^S_{s=1} q(\widetilde{X}_s;\beta),\\
  W&=\Big\{\{Y_t, A_t, X_t\}_{t=1}^{T}, \{\widetilde{X}_s\}_{s=1}^{S}\Big\}.
\end{align*}

Then, recall that in Appendix~\ref{appdx:submodel}, we define the score function as
\begin{align*}
    &S(W;\beta) \coloneqq \frac{\mathrm{d}}{\mathrm{d}\beta} \left(\sum^T_{t=1}\log \left(f(Y_t, A_t, X_t; \beta)\right) + \sum^S_{s=1}\log\big(q(\widetilde{X}_s; \beta)\big)\right)\\
    &= \sum^T_{t=1}\left(\mathbbm{1}[A_t=1] \left(s^1(Y_t\mid X_t; \beta) + \frac{\dot{w}(1\mid X_t;\beta)}{w(1\mid X_t; \beta)}\right) + \mathbbm{1}[A_t=0] \left(s^0(Y_t\mid X_t; \beta) + \frac{\dot{w}(0\mid X_t;\beta)}{w(0\mid X_t; \beta)}\right) + \zeta(X_t; \beta)\right)\\
    &\ \ \ \ \ \ \ + \sum^S_{s=1}\widetilde{\zeta}(\widetilde{X}_s; \beta),
\end{align*}
where recall that $s^d$, $\dot{\omega}$, $\zeta$, and $\widetilde{\zeta}$ are defined as
\begin{align*}
    s^d(y\mid x; \beta) &\coloneqq \frac{\mathrm{d}}{\mathrm{d}\beta} \log r^d(y\mid x; \beta),\\
    \dot{\omega}(d\mid x; \beta) &\coloneqq \frac{\mathrm{d}}{\mathrm{d}\beta} w(d\mid x; \beta),\\
    \zeta(x; \beta) &\coloneqq \frac{\mathrm{d}}{\mathrm{d}\beta} \log p(x; \beta),\\
    \widetilde{\zeta}(\widetilde{x}; \beta) &\coloneqq \frac{\mathrm{d}}{\mathrm{d}\beta} \log q(\widetilde{x}; \beta).
\end{align*}

As shown in Lemma~\ref{lem:deriv}, we have
\begin{align*}
    &\frac{\partial \theta(\beta)}{\partial \beta} = \mathbb{E}_{W}\Big[Y(1)s^1(Y(1)\mid X; \beta) - Y(0)s^0(Y(0)\mid X; \beta) +  \theta(\widetilde{X}; \beta)\widetilde{\zeta}(\widetilde{X}; \beta)\Big],
\end{align*}
where recall that
\begin{align*}
        \theta(X; \beta) &\coloneqq \mathbb{E}_{(Y, A)\sim f(y, a\mid X; \beta)}[Y(1)\mid X] - \mathbb{E}_{(Y, A)\sim f(y, a\mid X; \beta)}[Y(0)\mid X].
\end{align*}

Then, based on the above preparation, we prove \cref{lem:basic_bound} by guess and verify.

\begin{proof}[Proof of \cref{lem:basic_bound}]
Let us denote $\mathbb{E}_{\widetilde{X}\sim q(x), (Y, A)\sim f(y, a\mid X)}$ by $\mathbb{E}_W$.
    We first verify that $\mathbb{E}_W[\psi^2(W)] < \infty$. Here, we have 
    \begin{align*}
        \mathbb{E}_W\left[\left(\frac{\mathbbm{1}[A = 1]\big( Y - \mathbb{E}[Y(1)\mid \widetilde{X}]\big)}{w(1\mid \widetilde{X})} - \frac{\mathbbm{1}[A = 0]\big( Y - \mathbb{E}[Y(0)\mid \widetilde{X}]\big)}{w(0\mid \widetilde{X})}\right)\frac{q(\widetilde{X})}{p(\widetilde{X})}\mid \widetilde{X} \right] = 0.
    \end{align*}
    Therefore, we have
    \begin{align*}
        &\mathbb{E}_W\left[\left(\left(\frac{\mathbbm{1}[A = 1]\big( Y - \mathbb{E}[Y(1)\mid \widetilde{X}]\big)}{w(1\mid \widetilde{X})} - \frac{\mathbbm{1}[A = 0]\big( Y - \mathbb{E}[Y(0)\mid \widetilde{X}]\big)}{w(0\mid \widetilde{X})}\right)\frac{q(\widetilde{X})}{p(\widetilde{X})}\right)\left(\theta_0(\widetilde{X}) - \theta_0\right)\right]\\
        &=\mathbb{E}_W\left[\left(\theta_0(\widetilde{X}) 
 - \theta_0\right) \mathbb{E}_W\left[\left(\frac{\mathbbm{1}[A = 1]\big( Y - \mathbb{E}[Y(1)\mid \widetilde{X}]\big)}{w(1\mid \widetilde{X})} - \frac{\mathbbm{1}[A = 0]\big( Y - \mathbb{E}[Y(0)\mid \widetilde{X}]\big)}{w(0\mid X)}\right)\frac{q(\widetilde{X})}{p(\widetilde{X})}\mid \widetilde{X}\right]\right]\\
        &= 0.
    \end{align*}
     Hence, we have
    \begin{align*}
        &\mathbb{E}_W\left[\psi^2(W)\right]\\
        &= \mathbb{E}_W\Bigg[\Bigg(\frac{1}{T}\sum^T_{t=1}\left(\frac{\mathbbm{1}[A_t = 1]\big( Y_t - \mathbb{E}[Y_t(1)\mid \widetilde{X}_t]\big)}{w(1\mid \widetilde{X}_t)} - \frac{\mathbbm{1}[A_t = 0]\big( Y_t - \mathbb{E}[Y_t(0)\mid \widetilde{X}_t]\big)}{w(0\mid \widetilde{X}_t)}\right)\frac{q(\widetilde{X}_t)}{p(\widetilde{X}_t)}\\
        &\ \ \ \ \ \ \ \ \ \ \ \ \ \ \ \ \ \ \ \ \ \ \ \ \ \ \ \ \ \ \ \ \ \ \ \ \ \ \ \ \ \ \ \ \ \ \ \ \ \ \ \ \ \ \ \ \ \ \ \ \ \ \ \ \ \ \ \ \ \ \ \ \ \ \ \ \ \ \ \ \ \ \ \ \ \ \ \ \ \ \ \ \ \ \ \ \ \ \ \ \ \ \ \ \ \ \ \ \ \ \ \ \ \ \ \ \ \ \ \ \ \ \ \ \ \ + \frac{1}{S}\sum^S_{s=1}\left(\theta_0(\widetilde{X}_s) - \theta_0\right)\Bigg)^2\Bigg]\\
        &= \mathbb{E}_W\left[\frac{1}{T}\left(\frac{\sigma^2_0(1)(\widetilde{X}_t)}{w(1\mid \widetilde{X}_t)} + \frac{\sigma^2_0(0)\widetilde{X}_t)}{w(0\mid \widetilde{X}_t)}\right)\frac{q^2(X_t)}{p^2(\widetilde{X}_t)} + \frac{1}{S}\Big\{\theta_0(\widetilde{X}_s) - \theta_0\Big\}^2\right] < \infty.
    \end{align*}

We next verify that \eqref{eq:riesz} holds ($\frac{\partial \theta(\beta_0)}{\partial \beta} \coloneqq \mathbb{E}_W\Big[\psi(W)S(W; \beta_0)\Big]$) by showing 
\begin{align}
\label{eq:target_lower1}
    &\mathbb{E}_W\Big[\psi(W)S(W; \beta_0)\Big]= \mathbb{E}_W\Big[Y(1)s^1(Y(1)\mid X; \beta_0) - Y(0)s^0(Y(0)\mid X; \beta_0) +  \theta(\widetilde{X}; \beta_0)\widetilde{\zeta}(\widetilde{X}; \beta_0)\Big],
\end{align}
from Lemma~\ref{lem:deriv}.

Under our choice of $\psi(W)$, the LHS of \eqref{eq:target_lower1} becomes 
\begin{align*}
    &\mathbb{E}_W\Big[\psi(W)S(W; \beta)\Big]\\
    &=\mathbb{E}_W\Bigg[\Bigg(\frac{1}{T}\sum^T_{t=1}\left(\frac{\mathbbm{1}[A_t = 1]\big( Y_t - \mathbb{E}[Y_t(1)\mid X]\big)}{w(1\mid X_t)} - \frac{\mathbbm{1}[A_t = 0]\big( Y_t - \mathbb{E}[Y_t(0)\mid X_t]\big)}{w(0\mid X_t)}\right)\frac{q(X_t)}{p(X_t)}\\
    &\ \ \ \ \ \ \ \ \ \ \ \ \ \ \ \ \ \ \ \ \ \ \ \ \ \ \ \ \ \ \ \ \ \ \ \ \ \ \ \ \ \ \ \ \ \ \ \ \ \ \ \ \ \ \ \ \ \ \ \ \ \ \ \ \ \ \ \ \ \ \ \ \ \ \ \ \ \ \ \ \ \ \ \ \ \ \ \ \ \ \ \ \ \ \ \ \ \ \ \ \ \ \ \ \ \ \ \ \ \ \ \ \ \ \ \ \ \ \ \ \ \ \ \ \ \ + \frac{1}{S}\sum^S_{s=1}\Big\{\theta_0(\widetilde{X}_s) - \theta_0\Big\}\Bigg)\\
    &\times \Bigg(\sum^T_{t=1}\left(\mathbbm{1}[A_t=1] \left(s^1(Y_t(1)\mid X_t; \beta) + \frac{\dot{w}(1\mid X;\beta)}{w(1\mid X_t; \beta)}\right) + \mathbbm{1}[A_t=0] \left(s^0(Y_t(0)\mid X_t; \beta) + \frac{\dot{w}(0\mid X_t;\beta)}{w(0\mid X_t; \beta)}\right)\right)\\
    &\ \ \ \ \ \ \ \ \ \ \ \ \ \ \ \ \ \ \ \ \ \ \ \ \ \ \ \ \ \ \ \ \ \ \ \ \ \ \ \ \ \ \ \ \ \ \ \ \ \ \ \ \ \ \ \ \ \ \ \ \ \ \ \ \ \ \ \ \ \ \ \ \ \ \ \ \ \ \ \ \ \ \ \ \ \ \ \ \ \ \ \ \ \ \ \ \ \ \ \ \ \ \ \ \ \ \ \ \ \ \ \ \ \ \ \ \ \ \ \ \ \ \ \ \ \ + \sum^S_{s=1}\widetilde{\zeta}(\widetilde{X}_s; \beta)\Bigg)\Bigg]\\
    &=\mathbb{E}_W\Bigg[\frac{1}{T}\sum^T_{t=1}\frac{\mathbbm{1}[A_t = 1]\big( Y_t - \mathbb{E}[Y_t(1)\mid X_t]\big)}{w(1\mid X_t)}\frac{q(X_t)}{p(X_t)}\left(s^1(Y_t(1)\mid X_t; \beta) + \frac{\dot{w}(1\mid X_t;\beta)}{w(1\mid X_t; \beta)}\right)\\
    &\ \ \ \ \ \ \ \ \ \ \ \ \ \ \ \ \ \ \ \ + \frac{1}{S}\sum^S_{s=1}\theta_0(\widetilde{X}_s)\sum^T_{t=1}\mathbbm{1}[A_t = 1]\left(s^1(Y_t(1)\mid X_t; \beta) + \frac{\dot{w}(1\mid X_t;\beta)}{w(1\mid X_t; \beta)}\right)\Bigg]\\
    &\ \ \ - \mathbb{E}_W\Bigg[\frac{1}{T}\sum^T_{t=1}\frac{\mathbbm{1}[A_t = 0]\big( Y_t - \mathbb{E}[Y_t(0)\mid X]\big)}{w(0\mid X_t)}\frac{q(X)}{p(X_t)}\left(s^0(Y_t(1)\mid X_t; \beta) + \frac{\dot{w}(0\mid X_t;\beta)}{w(0\mid X_t; \beta)}\right)\\
    &\ \ \ \ \ \ \ \ \ \ \ \ \ \ \ \ \ \ \ \ + \frac{1}{S}\sum^S_{s=1}\theta_0(\widetilde{X}_s)\sum^T_{t=1}\mathbbm{1}[A_t = 0]\left(s^0(Y_t(0)\mid X_t; \beta) + \frac{\dot{w}(0\mid X_t;\beta)}{w(0\mid X_t; \beta)}\right)\Bigg]\\
    &\ - \theta_0\mathbb{E}_W\Bigg[\frac{1}{T}\sum^T_{t=1}\left(\mathbbm{1}[A_t=1] \left(s^1(Y_t(1)\mid X_t; \beta) + \frac{\dot{w}(1\mid X_t;\beta)}{w(1\mid X_t; \beta)}\right) + \mathbbm{1}[A_t=0] \left(s^0(Y_t(0)\mid X_t; \beta) + \frac{\dot{w}(0\mid X_t;\beta)}{w(0\mid X_t; \beta)}\right)\right)\\
    &\ \ \ \ \ \ \ \ \ \ \ \ \ \ \ \ \ \ \ \ \ \ \ \ \ \ \ \ \ \ \ \ \ \ \ \ \ \ \ \ \ \ \ \ \ \ \ \ \ \ \ \ \ \ \ \ \ \ \ \ \ \ \ \ \ \ \ \ \ \ \ \ \ \ \ \ \ \ \ \ \ \ \ \ \ \ \ \ \ \ \ \ \ \ \ \ \ \ \ \ \ \ \ \ \ \ \ \ \ \ \ \ \ \ \ \ \ \ \ \ \ \ \ \ \ \ + \sum^S_{s=1}\widetilde{\zeta}(\widetilde{X}_s; \beta)\Bigg]\\
    &\ \ \ + \mathbb{E}_W\Bigg[\left(\frac{1}{T}\sum^T_{t=1}\left(\frac{\mathbbm{1}[A_t = 1]\big( Y_t - \mathbb{E}[Y_t(1)\mid X_t]\big)}{w(1\mid X_t)} - \frac{\mathbbm{1}[A_t = 0]\big( Y_t - \mathbb{E}[Y_t(0)\mid X_t]\big)}{w(0\mid X_t)}\right)\frac{q(X_t)}{p(X_t)}  + \frac{1}{S}\sum^S_{s=1}\theta_0(\widetilde{X}_s)\right)\\
    &\ \ \ \ \ \ \ \ \ \ \ \ \ \ \ \ \ \ \ \ \ \ \ \ \ \ \ \ \ \ \ \ \ \ \ \ \ \ \ \ \ \ \ \ \ \ \ \ \ \ \ \ \ \ \ \ \ \ \ \ \ \ \ \ \ \ \ \ \ \ \ \ \ \ \ \ \ \ \ \ \ \ \ \ \ \ \ \ \ \ \ \ \ \ \ \ \ \ \ \ \ \ \ \ \ \ \ \ \ \ \ \ \ \ \ \ \ \ \ \ \ \ \ \ \ \ \times \sum^S_{s=1}\widetilde{\zeta}(\widetilde{X}_s; \beta)\Bigg],
\end{align*}
where we used $\mathbbm{1}[A_t = 1]\mathbbm{1}[A_t = 0] = 0$. 

Here, it holds that
\begin{align*}
    &\mathbb{E}_W\Bigg[\mathbbm{1}[A = d]\left(s^d(Y(d)\mid X; \beta) + \frac{\dot{w}(d\mid X;\beta)}{w(d\mid X; \beta)}\right)\mid X\Bigg]\\
    &= \mathbb{E}_W\Bigg[\mathbbm{1}[A = d]\mid X\Bigg]\left(s^d(Y(d)\mid X; \beta) + \frac{\dot{w}(d\mid X;\beta)}{w(d\mid X; \beta)}\right)\\
    &= w(d\mid X; \beta)s^d(Y(d)\mid X; \beta) + \dot{w}(d\mid X;\beta)= 0,
\end{align*}
and 
\begin{align*}
    &\mathbb{E}_W\Bigg[\frac{1}{T}\sum^T_{t=1}\left(\frac{\mathbbm{1}[A_t = 1]\big( Y_t - \mathbb{E}[Y_t(1)\mid X_t]\big)}{w(1\mid X_t)} - \frac{\mathbbm{1}[A_t = 0]\big( Y_t - \mathbb{E}[Y_t(0)\mid X_t]\big)}{w(0\mid X_t)}\right)\frac{q(X_t)}{p(X_t)} + \frac{1}{S}\sum^S_{s=1}\left(\theta_0(\widetilde{X}_s) - \theta_0\right)\Bigg]\\
    &= 0.
\end{align*}
Therefore, we have
\begin{align*}
    &\mathbb{E}_W\Big[\psi(W)S(W; \beta_0)\Big]=\mathbb{E}_W\Bigg[Y(1)\frac{q(X)}{p(X)}s^1(Y(1)\mid X; \beta) - Y(0)\frac{q(X)}{p(X)}s^0(Y(1)\mid X; \beta) + \theta_0(\widetilde{X})\zeta(\widetilde{X}; \beta)\Bigg].
\end{align*}

Here, from Lemma~\ref{lem:deriv}, we have
\begin{align*}
    &\mathbb{E}_W\Bigg[Y(1)\frac{q(X)}{p(X)}s^1(Y(1)\mid X; \beta) - Y(0)\frac{q(X)}{p(X)}s^0(Y(1)\mid X; \beta) + \theta_0(\widetilde{X})\zeta(\widetilde{X}; \beta)\Bigg]\\
    &= \mathbb{E}_W\Big[Y(1)s^1(Y(1)\mid X; \beta) - Y(0)s^0(Y(0)\mid X; \beta) +  \theta(\widetilde{X}; \beta)\widetilde{\zeta}(\widetilde{X}; \beta)\Big].
\end{align*}

Lastly, we verify that $\psi \in\mathcal{T}$. We define the functions as
\begin{align*}
    s^d(y\mid x; \beta_0) &= \frac{1}{T}\frac{\big( y - \mathbb{E}[Y(d)\mid X = x]\big)}{w(d\mid x)},\\
    \dot{w}(d\mid X;\beta_0) &= 0,\\
    \zeta(x; \beta_0) &= 0,\\
    \widetilde{\zeta}(\widetilde{x}; \beta_0) &= \frac{1}{S}\Big\{\theta_0(\widetilde{x}) - \theta_0\Big\}.
\end{align*}

Then, $\psi \in \mathcal{T}$. The proof is complete.
\end{proof}

\section{Proof of Theorem~\ref{thm:covariate_prob}}

\begin{proof}
We solve the following constraint optimization problem:

\begin{align*}
\min_{z\in\mathcal{P}}&\ \int \left(\left(\frac{\sigma^2_0(1)(x)}{w(1\mid x)} + \frac{\sigma^2_0(0)(x)}{w(0\mid x)} \right)\frac{q^2(x)}{p(x)} \right)\mathrm{d}x\\
\mathrm{s.t.}&\ \int p(x) dx = 1.
\end{align*}

Then, by using the Lagrangian variables $\alpha\in (-\infty, +\infty)$, we consider its Lagrange functional
\begin{align*}
\mathcal{L}(z; \alpha) =& \int \left(\frac{\sigma^2_0(1)(x)}{w(1\mid x)} + \frac{\sigma^2_0(0)(x)}{w(0\mid x)}\right)\frac{q^2(x)}{p(x)}\mathrm{d}x + \alpha\left(\int p(x) \mathrm{d}x - 1\right).
\end{align*}
Then we consider the infimum of $\mathcal{L}$ with respect to $z$ and maximization of $\mathcal{L}$ with respect to $\alpha$, that is, $\inf_{p\in\mathcal{P}}\sup_{\alpha\in(-\infty, +\infty)}\mathcal{L}(p; \alpha)$. Next, we consider the dual problem defined as $\sup_{\alpha\in(-\infty, +\infty)}\inf_{z\in\mathcal{P}}\mathcal{L}(z; \alpha)$. Owing to the convexity of $\int \left(\frac{\sigma^2_0(1)(x)}{w(1\mid x)} + \frac{\sigma^2_0(0)(x)}{w(0\mid x)}\right)\frac{q^2(x)}{p(x)}dx$ for $p(\cdot)\in\mathbb{R}^+$, the following equality is obtained \citep{rockafellar2015convex}:
\begin{align*}
\inf_{p\in\mathcal{P}}\sup_{\alpha\in(0, +\infty)}\mathcal{L}(p; \alpha) = \sup_{\alpha\in(0, +\infty)}\inf_{p\in\mathcal{P}}\mathcal{L}(p; \alpha).
\end{align*}
Hence, we discuss the solution to the dual problem. Given $\alpha$, we apply the Euler-Lagrange equation for calculating $\inf_{p} \mathcal{L}(p;\alpha)$ \citep{gelfand2000calculus}. The minimizer $p^*$ satisfies the following equation:
\begin{align*}
-\left(\frac{\sigma^2_0(1)(x)}{w(1\mid x)} + \frac{\sigma^2_0(0)(x)}{w(0\mid x)}\right)\frac{q^2(x)}{p^{*2}(x)}  + \alpha = 0.
\end{align*}
Therefore, we obtain the following KKT condition:
\begin{align*}
&-\left(\frac{\sigma^2_0(1)(x)}{w(1\mid x)} + \frac{\sigma^2_0(0)(x)}{w(0\mid x)}\right)\frac{q^2(x)}{p^{*2}(x)}  + \alpha^* = 0,\\
&0 < p^*(x),\\
&\alpha^*\left(\int p(x) dx - 1\right) = 0.
\end{align*}
Then, we obtain the optimal solutions as follows: 
\begin{itemize}
\item if $\alpha^*=0$, we have $-\left(\frac{\sigma^2_0(1)(x)}{w(1\mid x)} + \frac{\sigma^2_0(0)(x)}{w(0\mid x)}\right)\frac{q^2(x)}{p^{*2}(x)} =0$.
In this case, for $0 < p^*(x)$, there is no feasible solution. 
\item if $\alpha^* \neq 0$, we have $-\left(\frac{\sigma^2_0(1)(x)}{w(1\mid x)} + \frac{\sigma^2_0(0)(x)}{w(0\mid x)} + \theta_0(x)\right)\frac{q^2(x)}{p^{*2}(x)} + \theta_0 + \alpha^* = 0$ and $\alpha^*\left(\int p(x) dx - 1\right) = 0$. From these results, we have
\begin{align*}
&p^*(x) = \sqrt{\frac{\frac{\sigma^2_0(1)(x)}{w(1\mid x)} + \frac{\sigma^2_0(0)(x)}{w(0\mid x)} + \theta_0(x)}{\alpha^* - \theta_0}}q(x),\\
&\alpha^* - \theta_0 = \left(\mathbb{E}_q\left[\sqrt{\frac{\sigma^2_0(1)(x)}{w(1\mid x)} + \frac{\sigma^2_0(0)(x)}{w(0\mid x)}+ \theta_0(x)}\right]\right)^2.
\end{align*}
\end{itemize}
Therefore, we have
\begin{align*}
&p^*(x) = \frac{\sqrt{\frac{\sigma^2_0(1)(x)}{w(1\mid x)} + \frac{\sigma^2_0(0)(x)}{w(0\mid x)}}}{\mathbb{E}_q\left[\sqrt{\frac{\sigma^2_0(1)(x)}{w(1\mid x)} + \frac{\sigma^2_0(0)(x)}{w(0\mid x)}}\right]}q(x).
\end{align*}
\end{proof}

\section{Proof of Theorem~\ref{thm:asymp_dist}}
\label{appdx:minimax_opt}
Let $\mathbb{E}$ be an expectation operator over $(Y_t, A_t, X_t)$ and $\widetilde{X}_s$; that is, \[\mathbb{E}_{P_t} = \mathbb{E}_{\widetilde{X}_s \sim q(x), X_t\sim \widehat{p}_t(x), A_t \sim \widehat{w}_t(a\mid X_t), Y_t(A_t) \sim r^{A_t}(y\mid X_t)}.\]

Let us define
\begin{align*}
\xi_t &= \Big(\psi_t(Y_t, A_t, X_t; \widehat{w}_t, \widehat{p}_t) - \theta_0\Big)/\sqrt{\tau^*}\\
&= \left(\left(\frac{\mathbbm{1}[A_t = 1]\big(Y_t - \widehat{\mu}_t(1)(X_t)\big)}{\widehat{w}_t(1\mid X_t)} - \frac{\mathbbm{1}[A_t = 0]\big(Y_t - \widehat{\mu}_t(0)(X_t)\big)}{\widehat{w}_t(0\mid X_t)}\right)\frac{q(X_t)}{\widehat{p}_t(X_t)}+ \mathbb{E}_{\widetilde{X}\sim q(x)}\left[\widehat{\theta}_t(\widetilde{X})\right] - \theta_0\right) / \sqrt{\tau^*},
\end{align*}
where
\begin{align*}
    \tau^* = \mathbb{E}_{X\sim p^*(x)}\left[\left(\frac{\sigma^2_0(1)(X)}{w^*(1\mid X)} + \frac{\sigma^2_0(0)(X)}{w^*(0\mid X)}\right)\frac{q^2(X)}{p^{*2}(X)}\right]. 
\end{align*}

First, the sequence $\{\xi_t\}^T_{t=1}$ is a martingale difference sequence (MDS). 
\begin{lemma}
\label{lem:mds}
    Under the AAS-AIPWIS experiment, $\mathbb{E}_{Y_t(a) \sim r^a(y\mid x), A_t \sim \widehat{w}_t(a\mid x), X_t\sim \widehat{p}_t(x)}[\xi_t|\mathcal{F}_{t-1}] = 0$ holds. 
\end{lemma}
\begin{proof}[Proof of Lemma~\ref{lem:mds}]
    For each $a\in\{1, 0\}$, we have
\begin{align*}
    &\mathbb{E}_{P_t}\left[\frac{\mathbbm{1}[A_t = a]\big(Y_t - \widehat{\mu}_t(a)(X_t)\big)}{\widehat{w}_t(a\mid X_t)}\mid X_t, \mathcal{F}_{t-1}\right]\\
    &=\frac{\mathbb{E}_{P_t}\left[\mathbbm{1}[A_t = a]\big(Y_{t}(a)- \widehat{\mu}_{t}(a)(X_t)\big)\mid X_t, \mathcal{F}_{t-1}\right]}{\widehat{w}_t(a\mid X_t)}\\
    &=\frac{\widehat{w}_t(a| X_t)\big(\mu(a)(X_t) - \widehat{\mu}_{t}(a)(X_t)\big)}{\widehat{w}_t(a| X_t)}\\
    &= \mu(a)(X_t) - \widehat{\mu}_{t}(a)(X_t).
\end{align*}
Then, we have
\begin{align*}
    &\mathbb{E}_{P_t}\left[\xi_t\mid X_t, \mathcal{F}_{t-1}\right]\\
    &= \mathbb{E}_{P_t}\left[\left(\psi_t(Y_t, A_t, X_t; \widehat{w}_t, \widehat{p}_t) - \theta_0 \right)/\sqrt{\tau^*}\mid X_t, \mathcal{F}_{t-1}\right]\\
    &= \left(\left(\mu(1)(X_t) - \mu(0)(X_t) - \widehat{\mu}_{t}(1)(X_t) + \widehat{\mu}_{t}(0)(X_t)\right)\frac{q(X_t)}{\widehat{p}_t(X_t)}  + \mathbb{E}_{\widetilde{X}\sim q(x)}\left[\widehat{\theta}_t(\widetilde{X})\right]- \theta_0 \right)/\sqrt{\tau^*}.
\end{align*}
Note that $\mathbb{E}_{P_t}\left[\widehat{\mu}_{t}(1)(X_t) - \widehat{\mu}_{t}(0)\mid \mathcal{F}_{t-1}\right] - \mathbb{E}_{P_t}\left[\widehat{\theta}_t(\widetilde{X})\right] = 0$. 
Therefore, 
\begin{align*}
    &\mathbb{E}_{P_t}\left[\xi_t\mid \mathcal{F}_{t-1}\right] = 0.
\end{align*}
\end{proof}

Then, we show Theorem~\ref{thm:asymp_dist} by using the martingale central limit theorem (CLT), which is given as follows.

\begin{proposition}[Central Limit Theorem for an MDS; from Proposition~7.9, p.~194, \citet{Hamilton1994}; also see \citet{White1994}]
\label{prp:marclt} 
Let $\{W_t\}^\infty_{t=1}$ be a scalar MDS with some random variable $W_t$. Let $\overline{W} = \frac{1}{T}\sum^T_{t=1}W_t$. Suppose that 
\begin{description}
\item[(a)] $\mathbb{E}[W^2_t] = \sigma^2_t$, a positive value with $\frac{1}{T}\sum^T_{t=1}\sigma^2_t\to\sigma^2$, a positive value; 
\item[(b)] $\mathbb{E}[|W_t|^r] < \infty$ for some $r>2$ and all $t\in\mathbb{N}$;
\item[(c)] $\frac{1}{T}\sum^{T}_{t=1}W^2_t\xrightarrow{p}\sigma^2$. 
\end{description}
Then $\sqrt{T}\overline{W}\xrightarrow{d}\mathcal{N}(\bm{0}, \sigma^2)$.
\end{proposition}

Thus, to apply the martingale CLT, it is enough to verify 
\begin{description}
\item[(a)] $\frac{1}{T}\sum^T_{t=1}\mathbb{E}_{Y_t(a) \sim r^a(y\mid x), A_t \sim \widehat{w}_t(a\mid x), X_t\sim \widehat{p}_t(x)}[\xi^2_t]\to 1$; 
\item[(b)] $\mathbb{E}_{Y_t(a) \sim r^a(y\mid x), A_t \sim \widehat{w}_t(a\mid x), X_t\sim \widehat{p}_t(x)}[|\xi_t|^r] < \infty$ for some $r>2$ and all $t\in\mathbb{N}$;
\item[(c)] $\frac{1}{T}\sum^{T}_{t=1}\xi^2_t\xrightarrow{p}1$. 
\end{description}

Under this proof strategy, we show Theorem~\ref{thm:asymp_dist} below. Our proof is inspired by \citet{Kato2020adaptive}. We denote $\mathbb{E}_{Y_t(a) \sim r^a(y\mid x), A_t \sim \widehat{w}_t(a\mid x), X_t\sim \widehat{p}_t(x)}$ by $\mathbb{E}_{P_t}$. In the proof, we use the following results.
\begin{definition}\label{dfn:uniint}[Uniform integrability, \citet{Hamilton1994}, p.~191] Let $W_t \in \mathbb{R}$ be a random variable with a probability measure $P$.  A sequence $\{W_t\}$  is said to be uniformly integrable if for every $\epsilon > 0$ there exists a number $c>0$ such that 
\begin{align*}
\mathbb{E}[|W_t|\cdot \mathbbm{1}[|W_t| \geq c]] < \epsilon
\end{align*}
for all $t$.
\end{definition}
\begin{proposition}[Sufficient conditions for uniform integrability; Proposition~7.7, p.~191. \citet{Hamilton1994}]\label{prp:suff_uniint} Let $W_t, Z_t \in\mathbb{R}$ be random variables. Let $P$ be a probability measure of $Z_t$. (a) Suppose there exist $r>1$ and $M<\infty$ such that $\mathbb{E}_{P}[|W_t|^r]<M$ for all $t$. Then $\{W_t\}$ is uniformly integrable. (b) Suppose there exist $r>1$ and $M < \infty$ such that $\mathbb{E}_{P}[|Z_t|^r]<M$ for all $t$. If $W_t = \sum^\infty_{j=-\infty}h_jZ_{t-j}$ with $\sum^\infty_{j=-\infty}|h_j|<\infty$, then $\{W_t\}$ is uniformly integrable.
\end{proposition}

\begin{proposition}[$L^r$ convergence theorem, p~165, \citet{loeve1977probability}]
\label{prp:lr_conv_theorem}
Let $W_t$ be a random variable with probability measure $P$ and $w$ be a constant. 
Let $0<r<\infty$, suppose that $\mathbb{E}\big[|W_t|^r\big] < \infty$ for all $t$ and that $W_t \xrightarrow{\mathrm{p}}z$ as $n\to \infty$. The following are equivalent: 
\begin{description}
\item{(i)} $W_t\to w$ in $L^r$ as $t\to\infty$;
\item{(ii)} $\mathbb{E}\big[|W_t|^r\big]\to \mathbb{E}_{P}\big[|w|^r\big] < \infty$ as $t\to\infty$; 
\item{(iii)} $\big\{|W_t|^r, t\geq 1\big\}$ is uniformly integrable.
\end{description}
\end{proposition}

Then, we prove Theorem~\ref{thm:asymp_dist} as follows.
\begin{proof}[Proof of Theorem~\ref{thm:asymp_dist}]
We check conditions (a)--(b) in the following Steps~1--3.

\paragraph{Step~1: check of condition~(a)}
Note that $\sqrt{\tau^*}$ is non-random variable. Hence, we focus on the expectation of 
\[\left(\frac{\mathbbm{1}[A_t = 1]\big(Y_t - \widehat{\mu}_t(1)(X_t)\big)}{\widehat{w}_t(1\mid X_t)} - \frac{\mathbbm{1}[A_t = 0]\big(Y_t - \widehat{\mu}_t(0)(X_t)\big)}{\widehat{w}_t(0\mid X_t)}\right)\frac{q(X_t)}{\widehat{p}_t(X_t)} + \mathbb{E}_{\widetilde{X}\sim q(x)}\left[\widehat{\theta}_t(\widetilde{X})\right] - \theta_0\]
conditioned on $\mathcal{F}_{t-1}$. 

Instead of $\sum^T_{t=1}\mathbb{E}_P[\xi2_t(P)]$, we first consider the convergence of $\sum^T_{t=1}\mathbb{E}_P[(xi^2_t(P)|\mathcal{F}_{t-1}] = \Omega_t(P)$; that is, we show $\Omega_t(P) - 1\xrightarrow{\mathrm{p}} 0$. Then, by using the $L^r$-convergence theorem (Proposition~\ref{prp:lr_conv_theorem}), we show $\sum^T_{t=1}\mathbb{E}_P[\xi^2_t(P)] - 1 \to 0$. 

The conditional expectation is computed as follows:
\begin{align}
&\mathbb{E}_{P_t}\left[\left(\left(\frac{\mathbbm{1}[A_t = 1]\big(Y_t - \widehat{\mu}_t(1)(X_t)\big)}{\widehat{w}_t(1\mid X_t)} - \frac{\mathbbm{1}[A_t = 0]\big(Y_t - \widehat{\mu}_t(0)(X_t)\big)}{\widehat{w}_t(0\mid X_t)}\right)\frac{q(X_t)}{\widehat{p}_t(X_t)} + \mathbb{E}_{\widetilde{X}\sim q(x)}\left[\widehat{\theta}_t(\widetilde{X})\right] - \theta_0\right)^2 \mid \mathcal{F}_{t-1}\right]\nonumber
\\
&= \mathbb{E}_{P_t}\Bigg[\left(\frac{\mathbbm{1}[A_t = 1]\big(Y_t - \widehat{\mu}_t(1)(X_t)\big)}{\widehat{w}_t(1\mid X_t)} - \frac{\mathbbm{1}[A_t = 0]\big(Y_t - \widehat{\mu}_t(0)(X_t)\big)}{\widehat{w}_t(0\mid X_t)}\right)^2\frac{q^2(X_t)}{\widehat{p}^2_t(X_t)} \mid \mathcal{F}_{t-1}\Bigg]\nonumber\\
&\ \ \ + 2\mathbb{E}_{P_t}\Bigg[\left(\frac{\mathbbm{1}[A_t = 1]\big(Y_t - \widehat{\mu}_t(1)(X_t)\big)}{\widehat{w}_t(1\mid X_t)} - \frac{\mathbbm{1}[A_t = 0]\big(Y_t - \widehat{\mu}_t(0)(X_t)\big)}{\widehat{w}_t(0\mid X_t)}\right)\frac{q(X_t)}{\widehat{p}_t(X_t)}\Bigg(\mathbb{E}_{\widetilde{X}\sim q(x)}\left[\widehat{\theta}_t(\widetilde{X})\right] - \theta_0\Bigg) \mid \mathcal{F}_{t-1}\Bigg]\nonumber\\
&\ \ \ + \mathbb{E}_{P_t}\Bigg[\Bigg(\mathbb{E}_{\widetilde{X}\sim q(x)}\left[\widehat{\theta}_t(\widetilde{X})\right] - \theta_0\Bigg)^2 \mid \mathcal{F}_{t-1}\Bigg]\nonumber\\
&= \mathbb{E}_{P_t}\Bigg[\Bigg(\frac{\mathbbm{1}[A_t = 1]\big(Y_t - \widehat{\mu}_t(1)(X_t)\big)^2}{\widehat{w}^2_t(1\mid X_t)} - \frac{\mathbbm{1}[A_t = 0]\big(Y_t - \widehat{\mu}_t(0)(X_t)\big)^2}{\widehat{w}^2_t(0\mid X_t)}\Bigg)\frac{q^2(X_t)}{\widehat{p}^2_t(X_t)} \mid \mathcal{F}_{t-1}\Bigg]\nonumber\\
&\ \ \ + 2\mathbb{E}_{P_t}\Bigg[\Bigg(\frac{\mathbbm{1}[A_t = 1]\big(Y_t - \widehat{\mu}_t(1)(X_t)\big)}{\widehat{w}_t(1\mid X_t)} - \frac{\mathbbm{1}[A_t = 0]\big(Y_t - \widehat{\mu}_t(0)(X_t)\big)}{\widehat{w}_t(0\mid X_t)}\Bigg)\frac{q(X_t)}{\widehat{p}_t(X_t)}\Bigg(\mathbb{E}_{\widetilde{X}\sim q(x)}\left[\widehat{\theta}_t(\widetilde{X})\right] - \theta_0\Bigg) \mid \mathcal{F}_{t-1}\Bigg]\nonumber\\
&\ \ \ + \mathbb{E}_{P_t}\Bigg[\Bigg(\mathbb{E}_{\widetilde{X}\sim q(x)}\left[\widehat{\theta}_t(\widetilde{X})\right] - \theta_0\Bigg)^2 \mid \mathcal{F}_{t-1}\Bigg]\nonumber\\
&= \mathbb{E}_{P_t}\Bigg[\Bigg(\frac{\big(Y_t(1) - \widehat{\mu}_t(1)(X_t)\big)^2}{\widehat{w}_t(1\mid X_t)} - \frac{\big(Y_t(0) - \widehat{\mu}_t(0)(X_t)\big)^2}{\widehat{w}_t(0\mid X_t)}\Bigg)\frac{q^2(X_t)}{\widehat{p}^2_t(X_t)} \mid \mathcal{F}_{t-1}\Bigg]\nonumber\\
&\ \ \ + 2\mathbb{E}_{P_t}\Bigg[\Bigg(\theta_0(X_t) - \widehat{\theta}_t(X_t)\Bigg)\frac{q(X_t)}{\widehat{p}_t(X_t)}\Bigg(\mathbb{E}_{\widetilde{X}\sim q(x)}\left[\widehat{\theta}_t(\widetilde{X})\right] - \theta_0\Bigg) \mid \mathcal{F}_{t-1}\Bigg]\nonumber\\
&\ \ \ + \mathbb{E}_{P_t}\Bigg[\Bigg(\mathbb{E}_{\widetilde{X}\sim q(x)}\left[\widehat{\theta}_t(\widetilde{X})\right] - \theta_0\Bigg)^2 \mid \mathcal{F}_{t-1}\Bigg]\nonumber.
\end{align}

From $\widehat{\mu}_t(a)(x)\xrightarrow{\mathrm{a.s}} \mu^a(P)(x)$ and $\widehat{w}_t(a|x)\xrightarrow{\mathrm{a.s.}} w^*(a|x)$, for all $a\in\{1, 0\}$ and any $x\in\mathcal{X}$, we have
\begin{align*}
&\left(\frac{\big(Y_t(1) - \widehat{\mu}_t(1)(x)\big)^2}{\widehat{w}_t(1\mid x)} - \frac{\big(Y_t(0) - \widehat{\mu}_t(0)(x)\big)^2}{\widehat{w}_t(0\mid x)}\right)\frac{q^2(x)}{\widehat{p}^2_t(x)}\\
&\ \ \ + 2\left(\theta_0(x) - \widehat{\theta}_t(x)\right)\frac{q(x)}{\widehat{p}_t(x)}\left(\mathbb{E}_{\widetilde{X}\sim q(x)}\left[\widehat{\theta}_t(\widetilde{X})\right] - \theta_0\right) + \left(\mathbb{E}_{\widetilde{X}\sim q(x)}\left[\widehat{\theta}_t(\widetilde{X})\right] - \theta_0\right)^2\\
&\xrightarrow{\mathrm{a.s}} \left(\frac{\big(Y_t(1) - \mu(1)(x)\big)^2}{w^*(1\mid x)} - \frac{\big(Y_t(0) - \mu(0)(x)\big)^2}{w^*(0\mid x)}\right)\frac{q^2(x)}{p^{*2}(x)}.
\end{align*}

Let us define
\begin{align*}
    W_t(x) &\coloneqq \left(\frac{\big(Y_t(1) - \widehat{\mu}_t(1)(x)\big)^2}{\widehat{w}_t(1\mid x)} - \frac{\big(Y_t(0) - \widehat{\mu}_t(0)(x)\big)^2}{\widehat{w}_t(0\mid x)}\right)\frac{q^2(x)}{\widehat{p}^2_t(x)}\\
&\ \ \ + 2\left(\theta_0(x) - \widehat{\theta}_t(x)\right)\frac{q(x)}{\widehat{p}_t(x)}\left(\mathbb{E}_{\widetilde{X}\sim q(x)}\left[\widehat{\theta}_t(\widetilde{X})\right] - \theta_0\right) + \left(\mathbb{E}_{\widetilde{X}\sim q(x)}\left[\widehat{\theta}_t(\widetilde{X})\right] - \theta_0\right)^2\\
&\ \ \ - \left(\frac{\big(Y_t(1) - \mu(1)(x)\big)^2}{w^*(1\mid x)} - \frac{\big(Y_t(0) - \mu(0)(x)\big)^2}{w^*(0\mid x)}\right)\frac{q^2(x)}{p^{*2}(x)} - \left(\mathbb{E}_{\widetilde{X}\sim q(x)}\left[\widehat{\theta}_t(\widetilde{X})\right] - \theta_0\right)^2.
\end{align*}

Therefore, we obtain
\begin{align*}
    \mathbb{E}_{P_t}[W_t(X_t)\mid X_t, \mathcal{F}_{t-1}] \xrightarrow{\mathrm{a.s.}} 0.
\end{align*}

Next, we show $\mathbb{E}_{P_t}[W_t(X_t)] \to 0$ by using the $L^r$-convergence theorem (Proposition~\ref{prp:lr_conv_theorem}). Note that 
\begin{align*}
     \mathbb{E}_{P_t}\left[\left(\frac{\big(Y_t(1) - \mu(1)(X_t)\big)^2}{w^*(1\mid X_t)} - \frac{\big(Y_t(0) - \mu(0)(X_t)\big)^2}{w^*(0\mid X_t)}\right)\frac{q^2(X_t)}{p^{*2}(X_t)}\right]  = \tau^*.
\end{align*}

To apply the $L^r$-convergence theorem, we check that $W_t(X_t)$ is uniformly integrable. Here, $\xi_t$ is conditionally sub-Gaussian. From Lemma~2.7.6 of \citet{vershynin2018high}, the squared value $\xi^2_t$ is conditionally sub-exponential. Therefore, $W_t(X_t)$ is a sum of the sub-exponential random variable. This implies that $W_t(X_t)$ is uniformly integrable. As a result, from the $L^r$-convergence theorem, $\mathbb{E}_{P_t}[W_t(X_t)] \to 0$ holds. 

Finally, $\mathbb{E}_{P_t}[W_t(X_t)] / \tau^* \to 0$ implies 
\begin{align*}
    \mathbb{E}_{P_t}[\xi^2_t] - 1 \to 0.
\end{align*}

\paragraph{Step~2: check of condition~(b)} From Assumption~\ref{asm:subgaussian}, $\xi^{a,b}_t$ is sub-Gaussian. When $\xi_t$ is sub-Gaussian, the condition holds from Proposition~2.5.2 (ii) of \citet{vershynin2018high}. 

\paragraph{Step~3: check of condition~(c)}
Let us define $u_t$ as
\begin{align*}
&u_t = \xi^2_t - \mathbb{E}\big[\xi^2_t \mid \mathcal{F}_{t-1}\big]\\
&=\Bigg(\left(\frac{\big(Y_t(1) - \widehat{\mu}_t(1)(x)\big)^2}{\widehat{w}_t(1\mid x)} - \frac{\big(Y_t(0) - \widehat{\mu}_t(0)(x)\big)^2}{\widehat{w}_t(0\mid x)}\right)\frac{q^2(x)}{\widehat{p}^2_t(x)}\\
&\ \ \ \ \ \ \ \ \ + 2\left(\theta_0(x) - \widehat{\theta}_t(x)\right)\frac{q(x)}{\widehat{p}_t(x)}\left(\mathbb{E}_{\widetilde{X}\sim q(x)}\left[\widehat{\theta}_t(\widetilde{X})\right] - \theta_0\right) + \left(\mathbb{E}_{\widetilde{X}\sim q(x)}\left[\widehat{\theta}_t(\widetilde{X})\right] - \theta_0\right)^2\Bigg) / \tau^*\\
&\ \ \ - \mathbb{E}_{P_t}\Bigg[\Bigg(\left(\frac{\big(Y_t(1) - \widehat{\mu}_t(1)(x)\big)^2}{\widehat{w}_t(1\mid x)} - \frac{\big(Y_t(0) - \widehat{\mu}_t(0)(x)\big)^2}{\widehat{w}_t(0\mid x)}\right)\frac{q^2(x)}{\widehat{p}^2_t(x)}\\
&\ \ \ \ \ \ \ \ \ + 2\left(\theta_0(x) - \widehat{\theta}_t(x)\right)\frac{q(x)}{\widehat{p}_t(x)}\left(\mathbb{E}_{\widetilde{X}\sim q(x)}\left[\widehat{\theta}_t(\widetilde{X})\right] - \theta_0\right) + \left(\mathbb{E}_{\widetilde{X}\sim q(x)}\left[\widehat{\theta}_t(\widetilde{X})\right] - \theta_0\right)^2\Bigg) / \tau^* \mid \mathcal{F}_{t-1}\Bigg].
\end{align*}
Note that $u_t$ is an MDS. 

From the boundedness of each variable in $u_t$, we can apply the weak law of large numbers for an MDS and obtain
\begin{align*}
&\frac{1}{T}\sum^T_{t=1}u_t = \frac{1}{T}\sum^T_{t=1}\Big(\xi^2_t - \mathbb{E}\big[\xi^2_t \mid \mathcal{F}_{t-1}\big]\Big)\xrightarrow{\mathrm{p}} 0.
\end{align*}

From Assumption~\ref{asm:subgaussian} and the result of Step~1, as well as the proof in Lemma~10 in \citet{hadad2019}, we obtain 
\begin{align*}
\frac{1}{T}\sum^T_{t=1}\xi^2_t - 1\xrightarrow{\mathrm{p}} 0.
\end{align*}
Then, from the $L^r$-convergence theorem, we have
\begin{align*}
 \frac{1}{T}\sum^T_{t=1}\mathbb{E}\big[\xi^2_t \mid \mathcal{F}_{t-1}\big] - 1\xrightarrow{\mathrm{p}} 0.
\end{align*}

As a conclusion, we obtain
\begin{align*}
&\frac{1}{T}\sum^T_{t=1}\xi^2_t - 1 = \frac{1}{T}\sum^T_{t=1}\left(\xi^2_t - \mathbb{E}\big[\xi^2_t \mid \mathcal{F}_{t-1}\big] + \mathbb{E}\big[\xi^2_t\mid \mathcal{F}_{t-1}\big] - 1\right)\xrightarrow{\mathrm{p}} 0.
\end{align*}

\paragraph{Conclusion} In Steps~1--3. we checked the conditions required for applying the martingale CLT hold. Therefore, from the martingale CLT, we obtain the statement.
\end{proof}

\section{Details of the Simulation Studies}
This section provides details of simulation studies. The simulation studies were conducted by using the MacBook Pro and a Linux ubuntu18 server with $64$ cores and 
$251$ GB of RAM.

\subsection{Details of Alternative Experiments}
We explain the details of alternative experiments:
\begin{description}
    \item[DRCT.] We sample experimental units from $q(x)$. In each odd round $t = 1, 3, 5,\dots$, we assign treatment $A_t = 1$. In each even round $t = 2, 4, 6,\dots$, we assign treatment $A_t = 0$. Because there is no selection bias, we estimate the ATE by using the sample average; that is, 
    \[\widehat{\theta}_T \coloneqq \frac{1}{\sum^T_{t=1}\mathbbm{1}[A_t = 1]}\sum^T_{t=1}\mathbbm{1}[A_t = 1]Y_t - \frac{1}{\sum^T_{t=1}\mathbbm{1}[A_t = 0]}\sum^T_{t=1}\mathbbm{1}[A_t = 0]Y_t.\]
    \item[RCT.] We sample experimental units from $q(x)$. In each round, we assign treatment $A_t = a$ with probability $w(a\mid x)$ for all $a\in\{1, 0\}$ and any $x\in\mathcal{X}$. Because there is no selection bias, we estimate the ATE by using the sample average; that is, 
    \[\widehat{\theta}_T \coloneqq \frac{1}{\sum^T_{t=1}\mathbbm{1}[A_t = 1]}\sum^T_{t=1}\mathbbm{1}[A_t = 1]Y_t - \frac{1}{\sum^T_{t=1}\mathbbm{1}[A_t = 0]}\sum^T_{t=1}\mathbbm{1}[A_t = 0]Y_t.\]
    \item[AS-AIPW (Oracle).] We sample experimental units from $q(x)$. In each round, we assign treatment $A_t = a$ with probability $w^*(a\mid x)$ for all $a\in\{1, 0\}$ and any $x\in\mathcal{X}$. Because there is selection bias, we estimate the ATE by using the AIPW estimator; that is, 
    \begin{align*}
        \widehat{\theta}_T &\coloneqq\frac{1}{T}\sum^T_{t=1}\frac{\mathbbm{1}[A_t = 1]\left(Y_t - \widehat{\mu}_t(1)(X_t)\right)}{w^*(1\mid X_t)} - \frac{1}{T}\sum^T_{t=1}\frac{\mathbbm{1}[A_t = 0]\left(Y_t - \widehat{\mu}_t(0)(X_t)\right)}{w^*(0\mid X_t)}\\
        &\ \ \ \ \ \ \ \ \ + \frac{1}{T}\sum^T_{t=1}\mathbb{E}_{X\sim q(x)}\left[\widehat{\mu}_t(1)(X) - \widehat{\mu}_t(0)(X)\right],
    \end{align*}
    where $\widehat{\mu}_t(a)(x)$ and $\mathbb{E}_{X\sim q(x)}\left[\widehat{\mu}_t(1)(X) - \widehat{\mu}_t(0)(X)\right]$ are constructed as well as the AAS-AIPWIW experiment. 
    \item[AAS-AIPWIW (Oracle).] We sample experimental units from the efficient covariate density $p^*(x)$. In each round, we assign treatment $A_t = a$ with probability $w^*(a\mid x)$ for all $a\in\{1, 0\}$ and any $x\in\mathcal{X}$. Then, we estimate the ATE by using the AIPWIW estimator; that is, 
    \begin{align*}
        \widehat{\theta}_T &\coloneqq \frac{1}{T}\sum^T_{t=1}\left(\frac{\mathbbm{1}[A_t = 1]\left(Y_t - \widehat{\mu}_t(1)(X_t)\right)}{w^*(1\mid X_t)} - \frac{\mathbbm{1}[A_t = 0]\left(Y_t - \widehat{\mu}_t(0)(X_t)\right)}{w^*(0\mid X_t)}\right)\frac{q(X_t)}{p^*(X_t)}\\
        &\ \ \ \ \ \ \ \ \ + \frac{1}{T}\sum^T_{t=1}\mathbb{E}_{X\sim q(x)}\left[\widehat{\mu}_t(1)(X) - \widehat{\mu}_t(0)(X)\right].
    \end{align*}
    \item[AS-AIPW.] We sample experimental units from $q(x)$. In each round, we assign treatment $A_t = a$ with probability $w_t(a\mid x)$ for all $a\in\{1, 0\}$ and any $x\in\mathcal{X}$, which is constructed by using past observations, as well as the AA-AIPWIW experiment. Then, we estimate the ATE by using the AIPW estimator; that is, 
    \begin{align*}
        \widehat{\theta}_T &\coloneqq \frac{1}{T}\sum^T_{t=1}\frac{\mathbbm{1}[A_t = 1]\left(Y_t - \widehat{\mu}_t(1)(X_t)\right)}{w_t(1\mid X_t)} - \frac{1}{T}\sum^T_{t=1}\frac{\mathbbm{1}[A_t = 0]\left(Y_t - \widehat{\mu}_t(0)(X_t)\right)}{w_t(0\mid X_t)}\\
        &\ \ \ \ \ \ \ \ \  + \frac{1}{T}\sum^T_{t=1}\mathbb{E}_{X\sim q(x)}\left[\widehat{\mu}_t(1)(X) - \widehat{\mu}_t(0)(X)\right].
    \end{align*}
\end{description}

\subsection{MSE with Various \texorpdfstring{$T$}{TEXT}}
In Figures~\ref{fig:timegaussian}--\ref{fig:timeuniform}, we plot the MSEs at $T = 500, 750, 1,000, 1,250, 1,500, 1,750, 2,000, 2,250, 2,500, 2,750, 3,000$. The AAS-AIPWIW experiment successfully reduces the MSEs. In some cases, the AAS-AIPWIW outperforms the AAS-AIPWIW (Oracle), which seems paradoxical since the former utilizes the true $p^*$. Such a phenomenon is often reported in the literature of semiparametric analysis.

\begin{figure}[t]
  \centering
    \includegraphics[width=100mm]{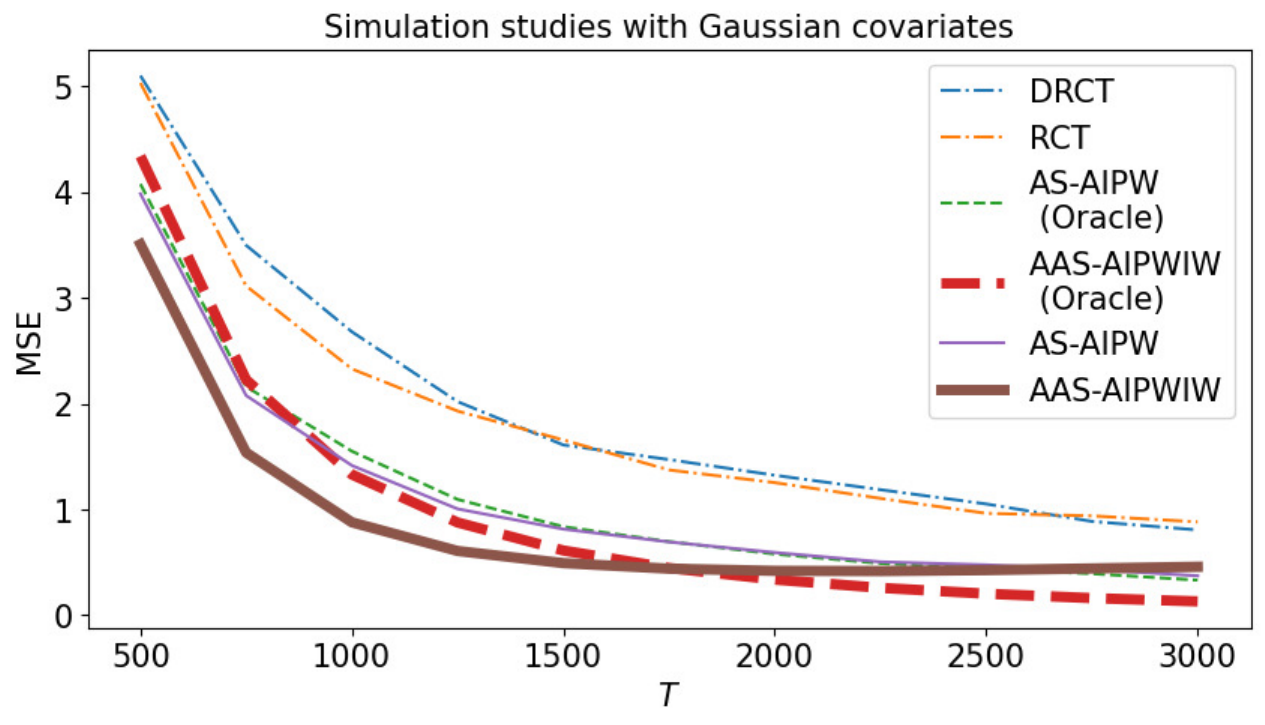}
    \vspace{-5mm}

\caption{MSEs at $T = 500, 750, 1,000, 1,250, 1,500, 1,750, 2,000, 2,250, 2,500, 2,750, 3,000$ in simulation studies with Gaussian covariates.}
\label{fig:timegaussian}
\vspace{5mm}
\centering
    \includegraphics[width=100mm]{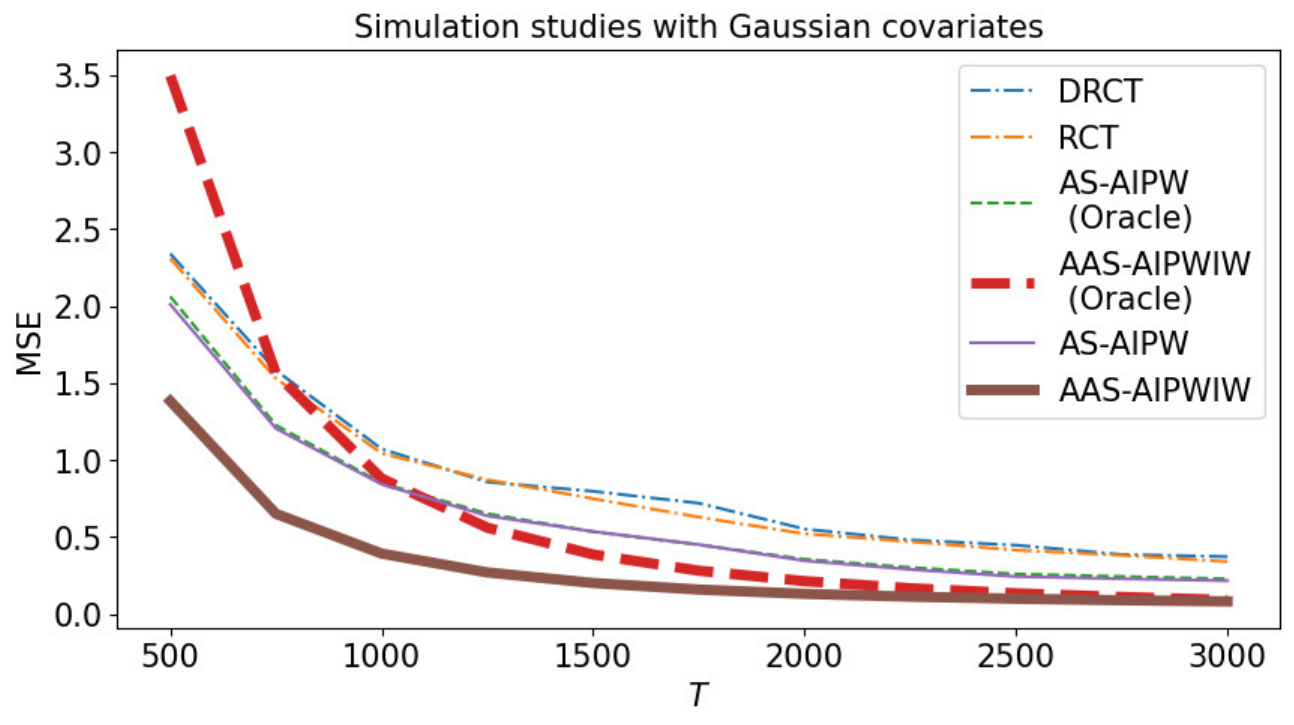}
    \vspace{-5mm}
\caption{MSEs at $T = 500, 750, 1,000, 1,250, 1,500, 1,750, 2,000, 2,250, 2,500, 2,750, 3,000$ in simulation studies with Uniform covariates.}
\label{fig:timeuniform}
\end{figure}

\subsection{Asymptotic Normality}
In Figures~\ref{fig:asympgaussian}--\ref{fig:asympuniform}, we plot empirical distributions of $\widehat{\theta}_T$ obtained via the RCT, and the AS-AIPW, and the AAS-AIPWIW experiments. While the RCT keeps a form of a Gaussian distribution, the shapes of empirical distributions of the AS-AIPW and the AAS-AIPWIW experiments seem to have some bias for the asymptotic distribution. From the theoretical viewpoint, such a bias will vanish as the sample size $T$ increases. 

\begin{figure}[t]
  \centering
    \includegraphics[width=100mm]{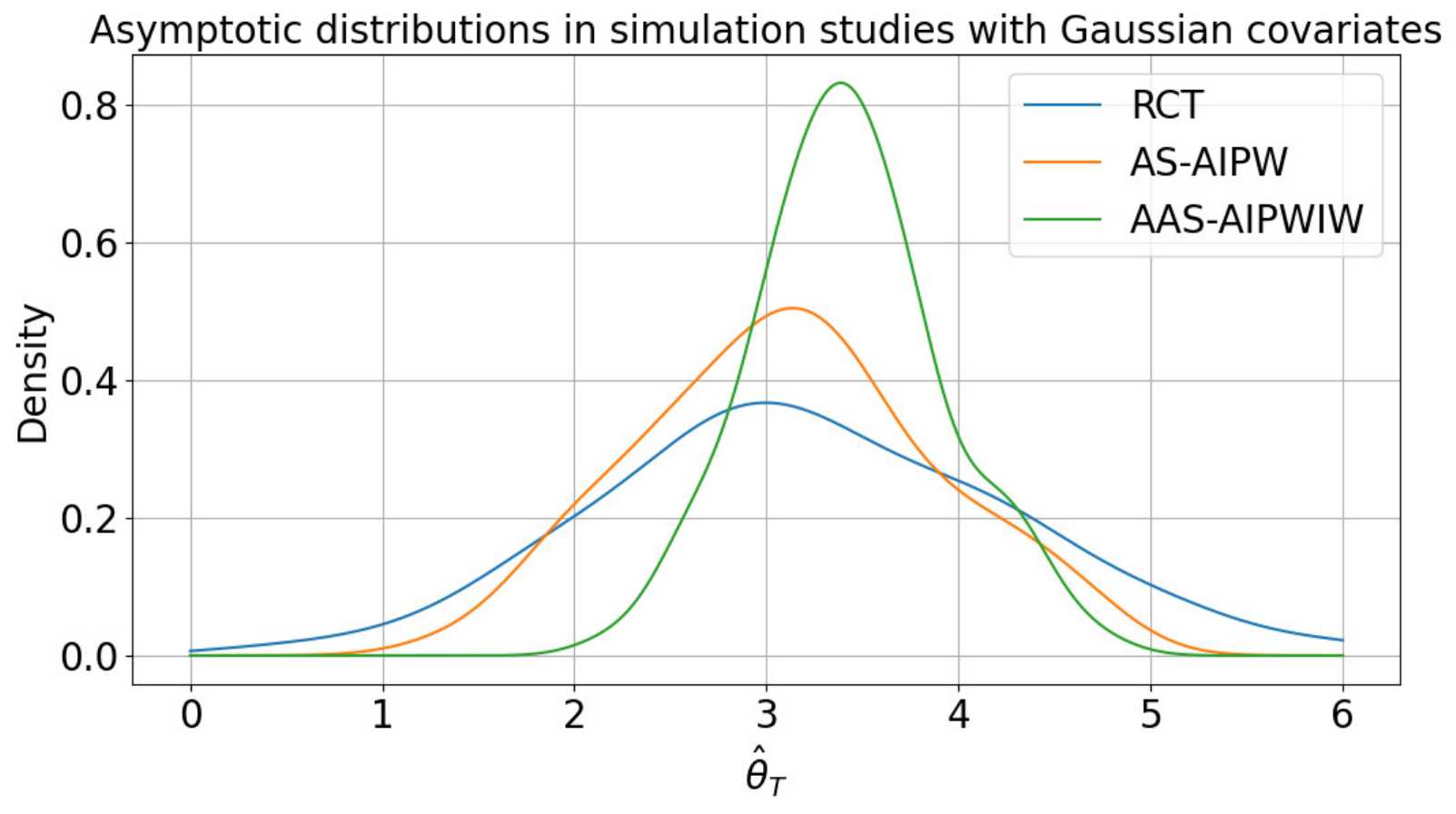}
    \vspace{-5mm}
\caption{Empirical distributions of ATE estimators in simulation studies with covariates following with a Gaussian distribution $\mathcal{N}(1, 25)$.}
\label{fig:asympgaussian}
\vspace{5mm}
\centering
    \includegraphics[width=100mm]{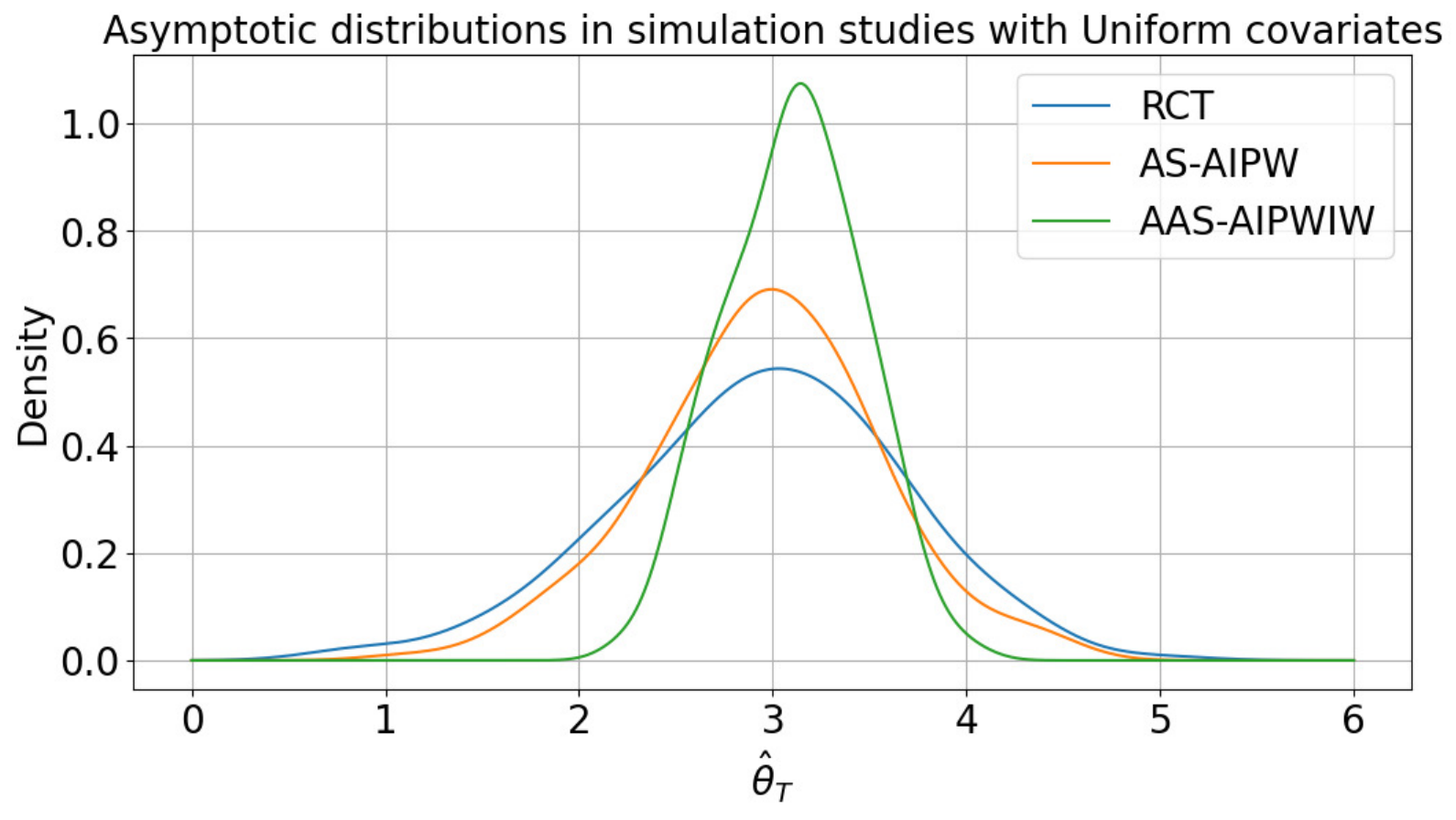}
    \vspace{-5mm}

\caption{Empirical distributions of ATE estimators in simulation studies with covariates following with a Uniform distribution $\mathcal{U}(-10, 10)$.}
\label{fig:asympuniform}
\end{figure}

\subsection{Coverage Ratio}
Based on the asymptotic normality, we construct confidence intervals of $\widehat{\theta}_T$ in the RCT, the AS-AIPW, and the AAS-AIPWIW experiments. Then, we confirm how many times $95\%$ confidence intervals include the true $\theta_0$. The ratio that the true $\theta_0$ drops into a confidence interval is referred to as a coverage ratio. 

As a result of simulation studies with $200$ iterations, we obtain coverage ratios for the RCT, AS-AIPW, and AAS-AIPWIW experiments with Gaussian and Uniform covariates. The results are shown in Table~\ref{tab:cov}.

\begin{table}[t]
    \centering
    \begin{tabular}{|c|c|c|c|}
    \hline
         &  RCT & AS-AIPW & AAS-AIPWIW\\
    \hline
      Gaussian covariates   & $95.5\%$ & $38.0\%$ & $26.5\%$\\
      Uniform covariates   & $95.5\%$ & $48.5\%$ & $44.0\%$\\
      \hline
    \end{tabular}
    \caption{Coverage ratios in simulation studies.}
    \label{tab:cov}
\end{table}

Although the RCT returns accurate confidence intervals, it seems that the confidence intervals of the AS-AIPW and the AAS-AIPWIW experiments have some bias. 

\subsection{Simulation Studies with Different Settings}
Lastly, we perform simulation studies with different settings.

\paragraph{Simulation studies with homogeneous variances} We investigate performances of the experiments in a setting where $\sigma_0(1)(x) = \sigma_0(0(x)$ for any $x\in\mathcal{X}$. We set 
\begin{align*}
    &\sigma^2_0(1)(x) = 2 + 1.2\mathrm{sin}(2x) +(x + x^2)/25,\\
    &\sigma^2_0(0)(x) = 2 + 1.2\mathrm{sin}(2x) +(x + x^2)/25.
\end{align*}
We show the results in Figures~\ref{fig:homovargaussian}--\ref{fig:homovaruniform}. 

Although our designed experiment still reduces the MSEs, the degree is less significant than in a case where variances are heterogeneous.

\paragraph{Simulation studies with homogeneous means}
Next, we investiate a setting where $\mu_0(1)(x) = \mu_0(0)(x)$ for any $x\in\mathcal{X}$. We define
\begin{align*}
    &\mu_0(1)(x) =  C_1x,\\
    &\mu_0(0)(x) = C_0x,
\end{align*}
where $C_1$ and $C_0$ are parameters so that $\mathbb{E}_{X\sim q(x)}[\mu_0(1)(X)] = \mathbb{E}_{X\sim q(x)}[Y(1)] = 10$ and $\mathbb{E}_{X\sim q(x)}[\mu_0(0)(X)] = \mathbb{E}_{X\sim q(x)}[Y(0)] = 7$. In Figures~\ref{fig:homomeangaussian}--\ref{fig:homomeanuniform}, we show the results. Our designed experiment still improves the performances compared to the RCTs, but the MSEs are almost the same as those in the AS-AIPW experiment. We consider that this is because $\mu_0(a)(x)$ has a simple form, and the conditional mean estimator in the AIPW estimator estimates it well, which makes the performances of the experiments insignificant. 

\begin{figure}[t]
  \centering
    \includegraphics[width=100mm]{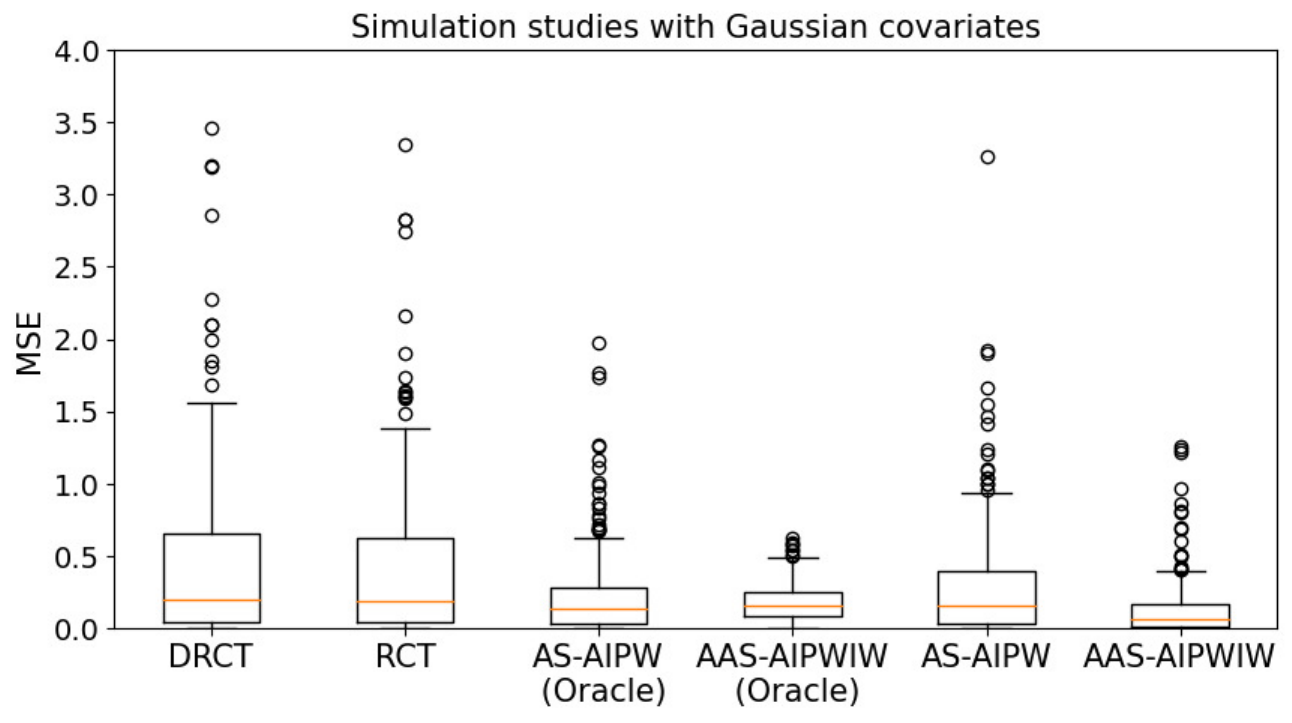}
\caption{Results of simulation studies with covariates following with a Gaussian distribution $\mathcal{N}(1, 25)$ and the homogeneous variances.}
\label{fig:homomeangaussian}
\vspace{5mm}
\centering
    \includegraphics[width=100mm]{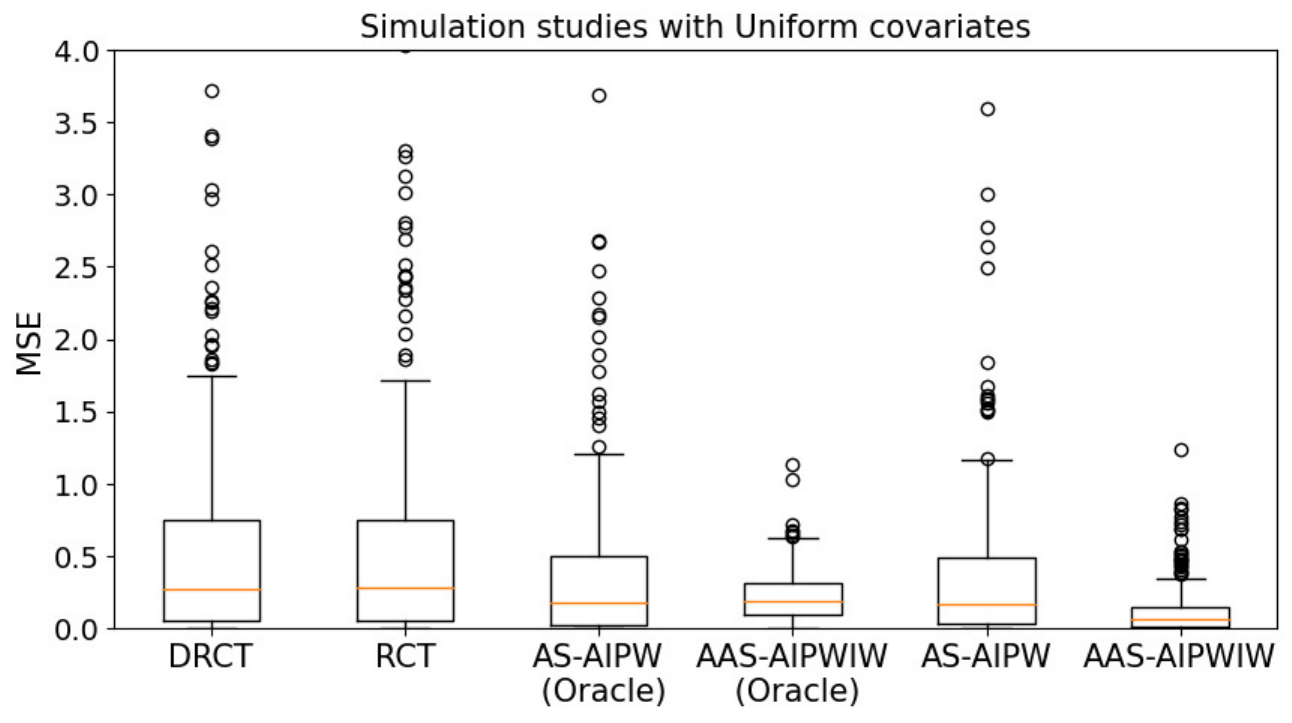}
\caption{Results of simulation studies with covariates following with a Uniform distribution $\mathcal{U}(-10, 10)$ and the homogeneous variances.}
\label{fig:homomeanuniform}
\end{figure}

\begin{figure}[t]
  \centering
    \includegraphics[width=100mm]{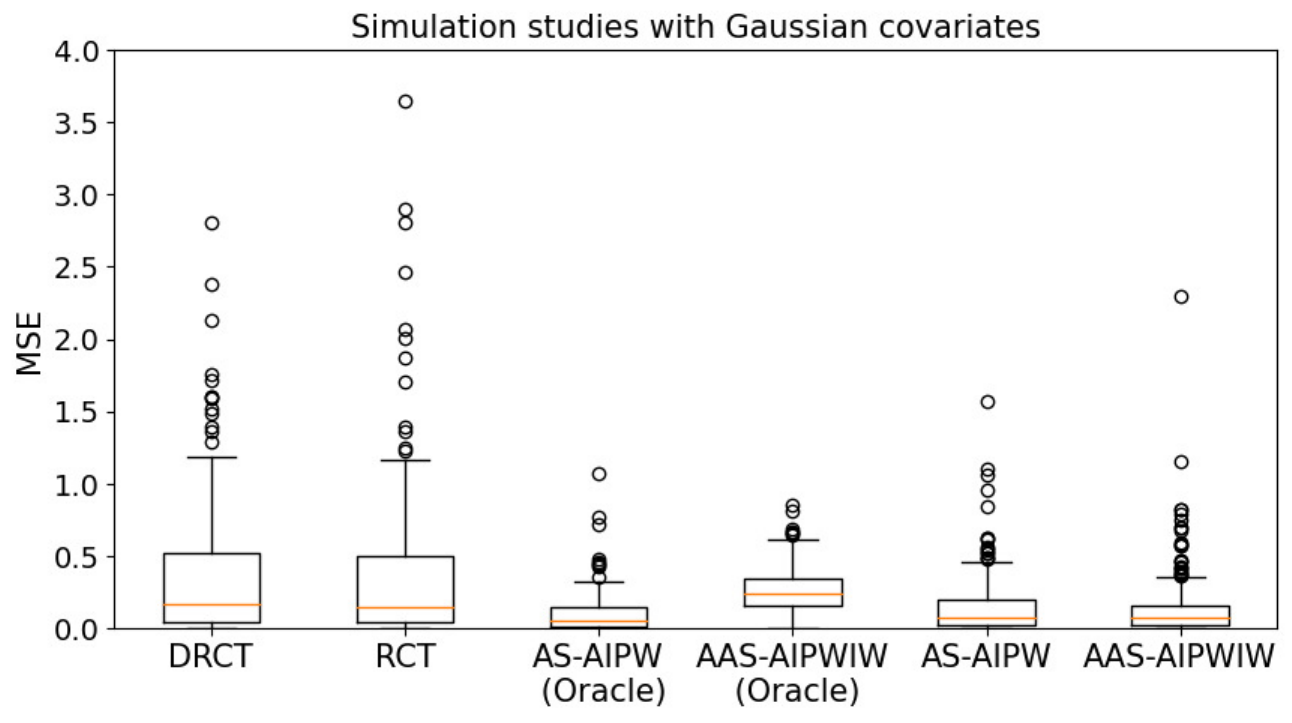}
\caption{Results of simulation studies with covariates following with a Gaussian distribution $\mathcal{N}(1, 25)$ and the homogeneous means.}
\label{fig:homovargaussian}
\centering
    \includegraphics[width=100mm]{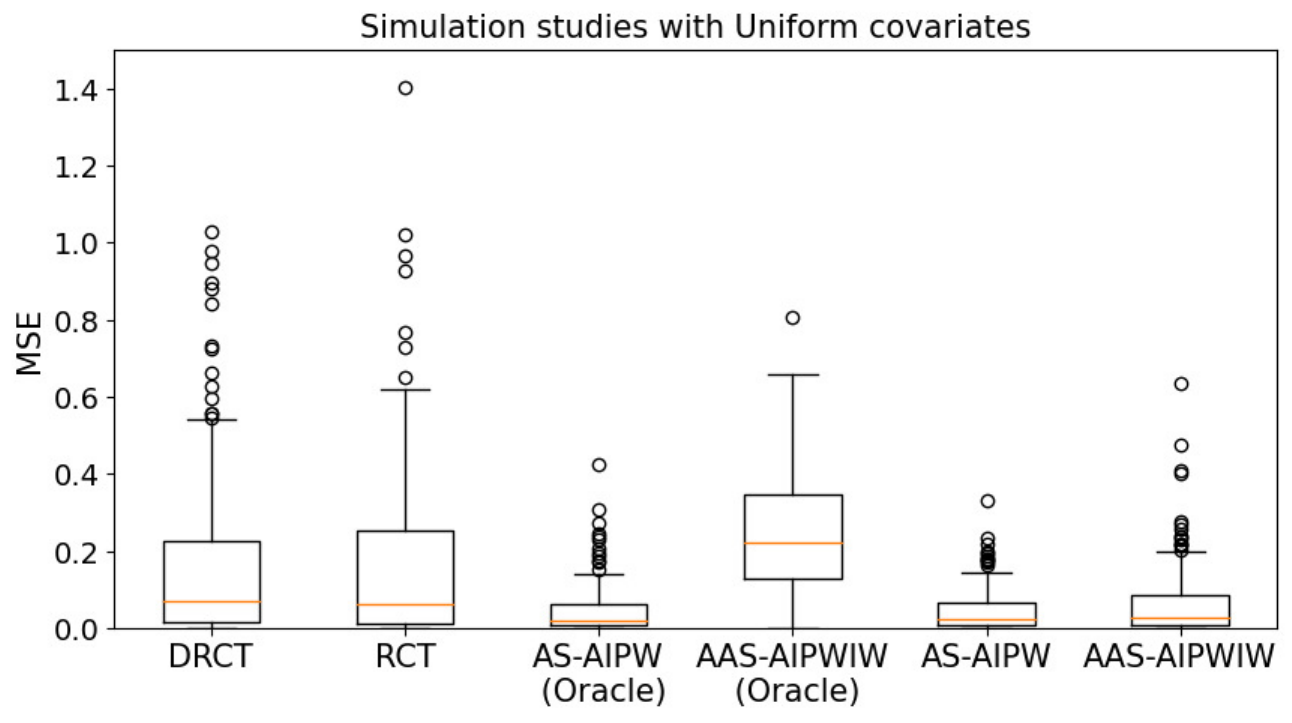}
\caption{Results of simulation studies with covariates following with a Uniform distribution $\mathcal{U}(-10, 10)$ and the homogeneous means.}
\label{fig:homovaruniform}
\end{figure}

\paragraph{Simulation studies with higher dimensional covariates}
Lastly, we investigate a case with higher dimensional covariates. In each round $t$, we generate $3$ and $20$ dimensional covariates. Let $\mathcal{N}(\bm{1}, \nu)$ be a multivariate normal distribution with the $d$-dimensional vector $\bm{1}=(1, 1, \ldots, 1)$ and a ($d\times d$)-dimensional diagonal matrix $\nu$ whose ($i, i$)-element is generated from a uniform distribution with support $(1, 5)$ in advance of the experiment. Let $X_t = (X_{1, t}, X_{2, t}, \ldots, X_{d, t})$ be the covariates in round $t$. 

As in the experiment in the main section, we specify the means as $\mathbb{E}_{X\sim q(x)}[Y(1)] = 10$ and $\mathbb{E}_{X\sim q(x)}[Y(0)] = 7$, under which the ATE is $\theta_0 = 3$. For any $X_t$, the conditional means are 
\begin{align*}
    \mu_0(1)(X_t) &\coloneqq  C_1\left(- X_{1,t} + 3X_{1,t}^2 - 1\right),\\
    \mu_0(0)(X_t) &\coloneqq C_0\left(0.1X_{1,t} + 0.2\right),
\end{align*}
where $C_1$ and $C_0$ are parameters so that $\mathbb{E}_{X\sim q(x)}[\mu_0(1)(X)] = \mathbb{E}_{X\sim q(x)}[Y(1)] = 10$ and $\mathbb{E}_{X\sim q(x)}[\mu_0(0)(X)] = \mathbb{E}_{X\sim q(x)}[Y(0)] = 7$. 

We define $\sigma^2_0(a)(X_t)$ as follows:
\begin{align*}
    \sigma^2_0(1)(X_t) &\coloneqq 2 + 1.2\sin(2X_{1,t}) + \frac{X_{1,t} + 2X^2_{1,t}}{25},\\
    \sigma^2_0(0)(X_t) &\coloneqq 2 + 0.8\cos\left(\frac{X_{1,t}}{2}\right) + \frac{X_{1,t}^2}{50}.
\end{align*}

Then, as in Figures~\ref{fig:gaussian_res}--\ref{fig:uniform_res}, we report the empirical MSEs for settings with each dimension $d=3, 10, 20$. The results are shown in Figure~\ref{fig:exp_high_dim}.

\begin{figure}[t]
  \centering
    \includegraphics[width=140mm]{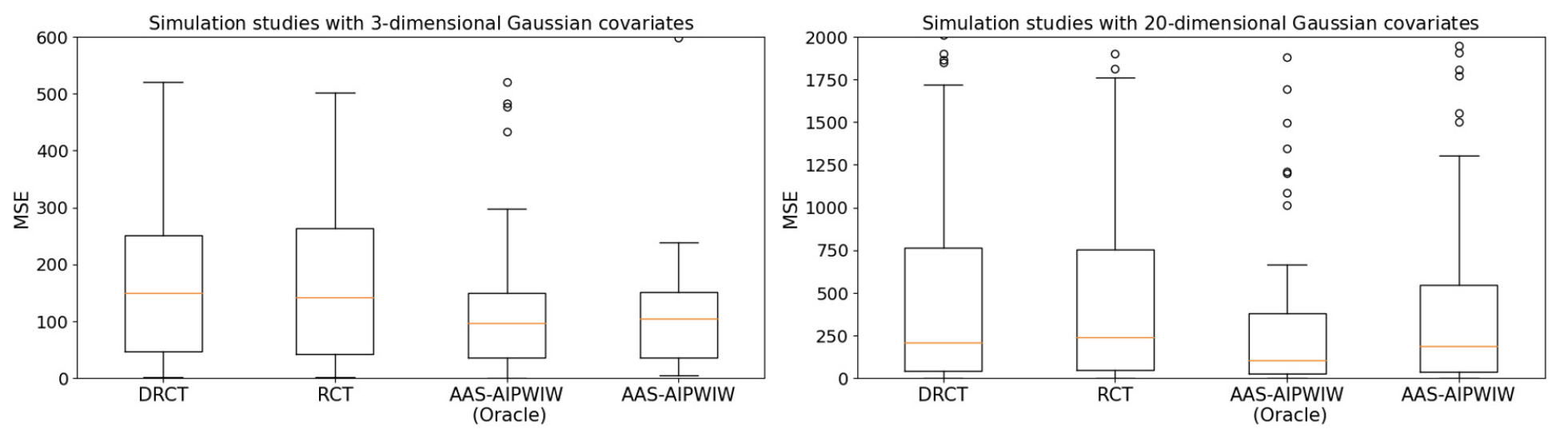}
\caption{Results of simulation studies with covariates following $3$-dimensional and $20$-dimensional Gaussian distributions.}
\label{fig:exp_high_dim}
\end{figure}

\subsection{Simulation Studies with the Semi-Synthetic Dataset}
This section provides simulation studies using the semi-synthetic dataset called the Infant Health and Development Program (IHDP). 

The IHDP dataset has simulated outcomes and real-world covariates \citep{Hill2011}. The total sample size is $747$, and the covariates consist of $6$ continuous variables and $19$ binary variables. 

The outcomes are generated artificially, while the covariates are real-world data. In \citep{Hill2011}, there are two scenarios called response surface A and response surface B. 
In response surface A, the potential outcomes $Y_t(1)$ and $Y_t(0)$ are generated as
\begin{align*}
    Y_t(0) &\sim \mathcal{N}(X^\top_{t}\bm{\gamma}_A, 1),\\
    Y_t(1) &\sim \mathcal{N}(X^\top_{t}\bm{\gamma}_A + 4, 1),
\end{align*}
where each element of $\bm{\gamma}_A\in\mathbb{R}^{25}$ is randomly generated from $\{0, 1, 2, 3, 4\}$ with probabilities $(0.5, 0.2, 0.15, 0.1, 0.05)$.  

In response surface B, the potential outcomes $Y_t(1)$ and $Y_t(0)$ are generated as
\begin{align*}
    Y_t(0) &\sim \mathcal{N}\left(\exp\left((X_{t} + W)^\top\bm{\gamma}_B\right), 1\right),\\
    Y_t(1) &\sim \mathcal{N}(X^\top_{t}\bm{\gamma}_B - q, 1),
\end{align*}
where $\bm{W}$ is an offset matrix of the same dimension as $X_t$ with every value equal to $0.5$, $q$ is a constant to normalize the average treatment effect conditional on $d=1$ to be $4$, and each element of $\bm{\gamma}_B\in\mathbb{R}^{25}$ is randomly generated from $\{0, 0.1, 0.2, 0.3, 0.4\}$ with probabilities $(0.6, 0.1, 0.1, 0.1, 0.1)$. 

Our experiment fixes $T=2000$. In each round $t$, we sample covariate $X_t$ with replacement from the $747$ samples. Then, following each setting of the response surfaces, we generate the corresponding outcomes. 

As in Figures~\ref{fig:gaussian_res}--\ref{fig:uniform_res}, we report the empirical MSEs for settings with each dimension $d=3, 10, 20$. The results are shown in Figure~\ref{fig:exp_res_ihdp}.

\begin{figure}[t]
  \centering
    \includegraphics[width=140mm]{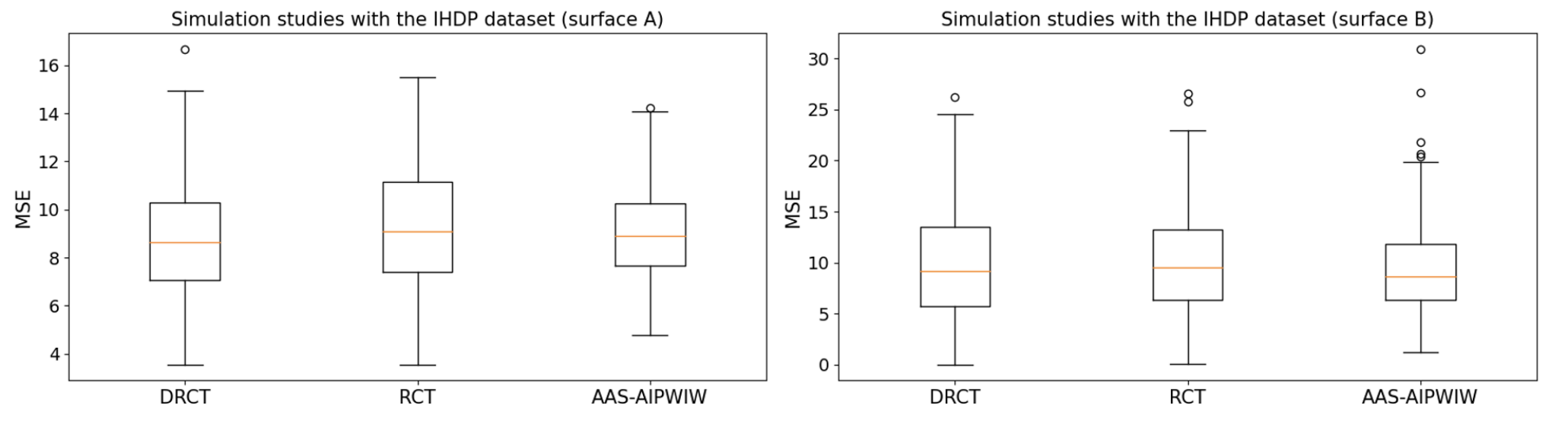}
\caption{Results of simulation studies with the IHDP dataset (response surface A and B).}
\label{fig:exp_res_ihdp}
\end{figure}

\section{Extensions}
\label{appdx:extension}
We introduce several extensions of our proposed framework.

\subsection{Rejection Sampling from \texorpdfstring{$q(x)$}{TEXT}}
\label{sec:rejection}
This section introduces a setting differing from that defined in Section~\ref{sec:problem}. 

\paragraph{Problem setting}
In this section, we assume that covariate $\widetilde{X}_t$ is generated from $q(x)$ in each round $t$, and we observe it. Then, we decide whether to assign a treatment to the sample. We employ the rejection sampling \citep{gelmanbda04} for this choice. If we accept $\widetilde{X}_t$, we assign treatment  $a \in \{1, 0\}$ and observe a corresponding outcome. If we do not accept $\widetilde{X}_t$, we do not assign any treatments and do not observe any outcomes. 
In this setting, we can remove Assumption~\ref{asm:covariate}, which assumes that $q(x)$ is known.
We illustrate the procedure in Figure~\ref{fig:rejection_sampling}. 

\begin{figure}[t]
  \centering
    \includegraphics[width=120mm]{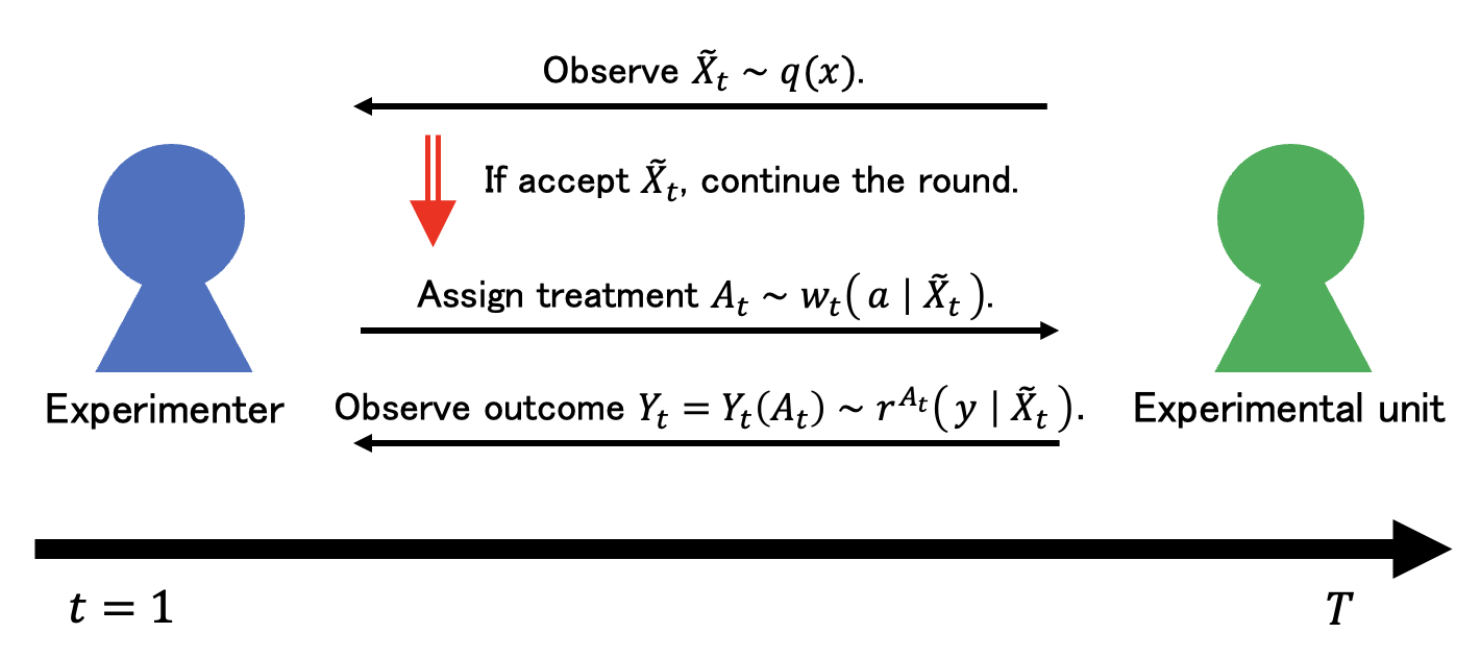}
\caption{Rejection sampling.}
\label{fig:rejection_sampling}
\end{figure}

\paragraph{Rejection sampling}
We first explain the rejection sampling. 
From assumptions, there exists an universal constant $M > 0$ such that $e_t(x) \coloneqq {\widehat{p}_t(x)}/{q(x)} < M$ for all $x\in\mathcal{X}$, where
\begin{align*}
    e_t(x) &= \frac{\sqrt{\big(\widehat{\sigma}_t(1)(x) + \widehat{\sigma}_t(0)(x)\big)^2}}{\frac{1}{t - 1}\sum^{t-1}_{s=1}\sqrt{\big(\widehat{\sigma}_t(1)(X_s) + \widehat{\sigma}_t(0)(X_s)\big)^2}}.
\end{align*}
holds. We accept $\widehat{X}_t$ if $U < e_t(\widehat{X}_t) / M$, and we assign a treatment to the sample. In contrast, we reject $\widehat{X}_t$ if $U \geq e_t(\widehat{X}_t) / M$, and we do not assign any treatments and cannot observe any outcomes.
Then, the following lemma holds.
\begin{lemma}
    Let $F_t(x)$ be a cumulative density function of $\widehat{p}_t(x)$. Let $U$ be a random variable generated from the uniform distribution with support $[0, 1]$. Then, it holds that
    \begin{align*}
        \mathbb{P}\Big(\widehat{X}_t \leq x\mid U < e_t(\widehat{X}_t) / M\Big) = F_t(x);
    \end{align*}
    that is, accepted $\widehat{X}_t$ follows a density $\widehat{p}_t(x)$.
\end{lemma}
We omit the proof. See \citet{gelfand2000calculus}.

\paragraph{Procedure}
By using this lemma, we consider the following procedure.
\begin{description}[topsep=0pt, itemsep=0pt, partopsep=0pt, leftmargin=*]
\item[Step~1:] based on past information, an experimenter decides  $p_t(x)$ and $w_t(a\mid x)$ for $a\in\{1, 0\}$ and $x\in\mathcal{X}$;
\item[Step~2:] an experimenter observes covariate $\widetilde{X}_t\sim q(x)$;
\item[Step~3:] by using the rejection sampling, we use $X_t = \widetilde{X}_t$ if $U_t \leq e_t(X_t) / M$, where $U_t$ is a uniform distribution with support $[0, 1]$;
\item[Step~3:] based on past information, we choose $A_t = a$ with probability $w_t(a\mid x)$;
\item[Step~4:] an experimenter observes outcome $Y_t=\mathbbm{1}[A_t = 1]Y_{t}(1) + \mathbbm{1}[A_t = 0]Y_{t}(0)$, where $Y_{t}(a)\sim r^a(y\mid x)$,
\item[Step~5:] after iterating Step~1--4 for $t=1,2,\dots, T$, in round $T$, we estimate the ATE using the estimator defined in \eqref{eq:AIPWIW2}.
\end{description}

\paragraph{ATE estimator}
The estimator is defined as
\begin{align}
\label{eq:AIPWIW2}
    \widehat{\theta}_T \coloneqq \frac{1}{T}\sum^T_{t=1}\widetilde{\Psi}_t(Y_t, A_t, X_t; w_t, p_t),
\end{align}
where
\begin{align*}
    &\widetilde{\Psi}_t(Y_t, A_t, X_t; \widehat{w}_t, \widehat{p}_t) \coloneqq \left(\frac{\mathbbm{1}[A_t = 1]\big(Y_t(1) - \widehat{\mu}_t(1)(X_t)\big)}{\widehat{w}_t(1\mid X_t)} - \frac{\mathbbm{1}[A_t = 0]\big(Y_t(0) - \widehat{\mu}_t(0)(X_t)\big)}{\widehat{w}_t(0\mid X_t)} + \widehat{\theta}_t(X_t)\right)\frac{1}{e_t(X_t)}.
\end{align*}
We can show that the estimator has the same asymptotic variance of $\widehat{\theta}_T$ in Theorem~\ref{thm:asymp_dist}. The result is summarized in the following proposition. Note that we do not use Assumption~\ref{asm:covariate}. We omit the proof since it is almost same as that for Theorem~\ref{thm:asymp_dist}.
\begin{proposition}
    Consider the AAS-AIPWIS experiment. 
    Suppose that Assumptions~\ref{asm:dist} and \ref{asm:subgaussian}--\ref{asm:nuisance_consistency} hold. Then, 
    \begin{align*}
        \sqrt{T}\left(\widehat{\theta}_T - \theta_0\right) \xrightarrow{\mathrm{d}}\mathcal{N}(0, \tau^*)
    \end{align*}
    holds as $T \to \infty$, where recall that $\tau^*$ is the semiparametric efficiency bound. 
\end{proposition}

\subsection{Batch Design}
We introduce a setting where each experimental unit visits sequentially. We can directly extend our experiment to a setting where sets of experimental units visit.

\subsection{Multiple Treatments}
We can consider a setting where there multiple treatments more $[K] \coloneqq \{1, 2, \dots, K\}$ for $K\geq 2$, instead of the binary treatments $\{1, 0\}$. In this case, it is unclear how we define the treatment effect. One way is to define the policy value by using a deterministic evaluation policy $\pi:[K]\times \mathcal{X} \to \mathbb{R}$. In many cases, it is assumed that $\sum_{a\in[K]}\pi(a\mid x)$ for any $x\in\mathcal{X}$, but we do not impose the restriction to derive a more general result, including a case with the restriction as a specific case. 

Under the evaluation policy, we define the policy value as 
$\theta_{0, K} \coloneqq \mathbb{E}\left[\sum_{a\in[K]}\pi(a\mid X)Y(a)\right]$. 

For $\theta_{0, K}$, the semiparametric efficiency bound is given as 
$\tau_K(w, p)= \mathbb{E}_{X\sim p(x)}\left[\sum_{a\in[K]}\frac{\sigma^2_0(a)(X)}{w(a\mid X)}\frac{q^2(X)}{p^2(X)}\right] = \mathbb{E}_{X\sim q(x)}\left[\sum_{a\in[K]}\frac{\sigma^2_0(a)(X)}{w(a\mid X)}\frac{q(X)}{p(X)}\right]$. 
Because the proof is almost the same as that for \cref{thm:lower}, we omit the proof. 

Then, the efficient probabilities are obtained as
\begin{align*}
&w^*(a\mid x) = \frac{\sigma_0(a)(x)}{\sum_{a\in[K]}\sigma_0(a)(x)}\\
&p^*(x) = \frac{\sqrt{\sum_{a\in[K]}\frac{\sigma^2_0(a)(x)}{w^*(a\mid x)}}}{\mathbb{E}_{X\sim q(x)}\left[\sqrt{\sum_{a\in[K]}\frac{\sigma^2_0(1)(X)}{w^*(1\mid X)}}\right]}q(x).
\end{align*}

We then just directly apply the AAS-AIPWIS experiment.

\end{document}